%% file: paper.tex
\documentclass[sigconf,screen]{acmart}
\acmConference[PL'18]{ACM SIGPLAN Conference on Programming Languages}{January 01--03, 2018}{New York, NY, USA}
\acmYear{2023}
\acmISBN{} % \acmISBN{978-x-xxxx-xxxx-x/YY/MM}
\acmDOI{} % \acmDOI{10.1145/nnnnnnn.nnnnnnn}
\startPage{1}

\setcopyright{none}

\usepackage[T1]{fontenc}
\usepackage{booktabs}   
\usepackage{subcaption}
\usepackage{graphicx}
\usepackage{adjustbox}
\usepackage{listings}

\usepackage{stmaryrd}
\usepackage{arydshln}
\usepackage{proof}
\usepackage{xspace}
\usepackage{multirow}
\usepackage{pifont}
\usepackage{enumitem}
\setlist{nosep,leftmargin=\parindent}

\usepackage{tikz}
\usepackage{tikz-qtree}
\usepackage[framemethod=tikz]{mdframed}
\usepackage{xcolor}
\usepackage{mathtools}
\usepackage{microtype}
\usepackage{stackengine}
\usepackage{hyperref}

\usepackage{wrapfig}

\usepackage{svg}

\newif\ifdispComments
\dispCommentsfalse

\input{defns.tex}

\usepackage{refs}
\begin{document}

%% Title information
\title[Program Synthesis is $\Sigma_3^0$-Complete]{Program Synthesis is $\Sigma_3^0$-Complete}         %% [Short Title] is optional;

                                        %% when present, will be used in
                                        %% header instead of Full Title.
%\titlenote{with title note}            %% \titlenote is optional;
                                        %% can be repeated if necessary;
                                        %% contents suppressed with 'anonymous'
%\subtitle{Subtitle}                     %% \subtitle is optional
%\subtitlenote{with subtitle note}       %% \subtitlenote is optional;
                                        %% can be repeated if necessary;
                                        %% contents suppressed with 'anonymous'

%% Author information
%% Contents and number of authors suppressed with 'anonymous'.
%% Each author should be introduced by \author, followed by
%% \authornote (optional), \orcid (optional), \affiliation, and
%% \email.
%% An author may have multiple affiliations and/or emails; repeat the
%% appropriate command.
%% Many elements are not rendered, but should be provided for metadata
%% extraction tools.

%% Author with single affiliation.
\author{Jinwoo Kim}                                        %% can be repeated if necessary
\affiliation{
  \institution{Seoul National University}            %% \institution is required
  \city{Seoul}
  \country{Republic of Korea}                  %% \country is recommended
}
\email{jinwoo.kim@sf.snu.ac.kr}          %% \email is recommended

%% Abstract
%% Note: \begin{abstract}...\end{abstract} environment must come
%% before \maketitle command
%\input{abstract}

%% Keywords
%% comma separated list
% \keywords{keyword1, keyword2, keyword3}  
%% \keywords are mandatory in final camera-ready submission

%% \maketitle
%% Note: \maketitle command must come after title commands, author
%% commands, abstract environment, Computing Classification System
%% environment and commands, and keywords command.

\input{abstract.tex}

\maketitle

\input{1intro.tex}

\input{2prelim.tex}

\input{4fo.tex}

\input{5completeness.tex}

\input{6variants.tex}

\input{7related.tex}

\input{8conclusion.tex}

%% Acknowledgments

%% Bibliography
%\bibliography{bibfile}
%\bibliographystyle{ACM-Reference-Format}
\bibliography{reference}

%% Appendix
% \newpage
% \appendix

% \input{appendix_proofs.tex}

\end{document}

%% file: defns.tex
\newcommand{\mypar}[1]{\noindent\textit{#1.}}
\newcommand{\myparvs}[1]{\vspace{1mm}\noindent\textit{{#1.}}}

\newcommand{\rone}{(\emph{i})~}
\newcommand{\rtwo}{(\emph{ii})~}

\newcommand{\lbbar}{\{\kern-0.5ex|}
\newcommand{\rbbar}{|\kern-0.5ex\}}

\newcommand{\biglbbar}{\left\{\kern-0.5ex\left|}
\newcommand{\bigrbbar}{\right|\kern-0.5ex\right\}}

\newcommand{\semantics}[1]{\llbracket {#1} \rrbracket}

\newcommand{\Etrue}{{{\mathsf{true}}}}
\newcommand{\Efalse}{{{\mathsf{false}}}}

\newcommand{\Eifthenelse}[3]{{\mathsf{if}}\ {#1}\ {\mathsf{then}}\ {#2}\ {\mathsf{else}}\ {#3}}
\newcommand{\Eif}[2]{{\mathsf{if}}\ {#1}\ {\mathsf{then}}\ {#2}}
\newcommand{\Eseq}[2]{{#1}{\mathsf{;}}\ {#2}}

\newcommand{\sem}[1]{\llbracket{#1}\rrbracket}       % [[ ]] 

\newcommand{\Eassign}[2]{{{#1}\ \mathsf{:=}\ {#2}}}

\newcommand{\Ewhile}[2]{\mathsf{while} \; {#1} \; \mathsf{do} \; {#2}}

\newcommand{\Omit}[1]{}

\usepackage[breakable]{tcolorbox}
\newenvironment{mybox}[1][gray!20]{
	\begin{tcolorbox}[   %% Adjust the following parameters at will.
		breakable,
		left=0pt,
		right=0pt,
		top=0pt,
		bottom=-1pt,
		colback=#1,
		colframe=#1,
		width=\dimexpr\textwidth\relax,
		boxsep=2pt,
		arc=0pt,outer arc=0pt,
		]
	}{
\end{tcolorbox}
}

\newcounter{resq}%[section]

\newcommand{\set}[1]{\{{#1}\}}

\newcommand{\lrangle}[1]{\langle {#1} \rangle}

\newcommand{\mymod}{\text{mod }}

\newcommand{\mymathcolor}[2]{\begingroup\color{#1}#2\endgroup}

\newcommand{\syconst}{\mathit{sy}_{\mathsf{id\_const}}}

\definecolor{dgreen}{RGB}{0,128,0}
\definecolor{dred}{RGB}{200,0,0}

\newcommand{\sic}[1]{\Sigma^0_{#1}}
\newcommand{\pic}[1]{\Pi^0_{#1}}

\newcommand{\gtgt}{G_{\text{IMP}}}
\newcommand{\ltgt}{\text{IMP}}
\newcommand{\Mifthenelse}[3]{{\mathit{if}}\ {#1}\ {\mathit{then}}\ {#2}\ {\mathit{else}}\ {#3}}
\newcommand{\sy}{\mathit{sy}}
\newcommand{\defeq}{\ensuremath{\triangleq}}
\newcommand{\bin}{\mathit{bin}}
\newcommand{\nullnt}{\mathit{NullNT}}
\newcommand{\nullop}{\mathsf{nop}}
\newcommand{\nullleaf}{\bullet}
\newcommand{\nullval}{\emptyset}
\newcommand{\h}{2^h - 2}

\newcommand{\stin}{\mathit{in}}
\newcommand{\stout}{\mathit{out}}

\newcommand{\fin}{\textnormal{FIN}}
\newcommand{\tot}{\textnormal{TOT}}
\newcommand{\cof}{\textnormal{COF}}
\newcommand{\syset}{\textnormal{SYNTH}}
\newcommand{\sysetfin}{\textnormal{SYNTH}_\textnormal{fin}}
\newcommand{\sysetnoloop}{\textnormal{SYNTH}_\textnormal{lf}}
\newcommand{\sysetpartial}{\textnormal{SYNTH}_\textnormal{part}}
\newcommand{\hset}{\mathsf{Halt}}

\newcommand{\godel}{Gödel\xspace}
\newcommand{\cone}{\mathit{c1}}
\newcommand{\ctwo}{\mathit{c2}}

\newcommand{\semcheckl}{\mathsf{check}}
\newcommand{\leafcheckl}{\mathsf{check\_leaf}}

\newcommand{\ptree}{\mathbf{p}}
\newcommand{\vtree}{\mathbf{v}}
\newcommand{\vztree}{\mathbf{v_z}}
\newcommand{\ftree}{\mathbf{f}}
\newcommand{\fqtree}{\mathbf{f_q}}
\newcommand{\gtree}{\mathbf{g}}

\newcommand{\semout}{\mathsf{output}}
\newcommand{\decode}{\mathsf{decode}}

\newcommand{\syc}{\sy_C}
\newcommand{\terminate}{\mathsf{Terminate}}
\newcommand{\fpair}{\mathsf{Pair}}

\def\sygus{\textsc{SyGuS}\xspace}

%% file: abstract.tex
\begin{abstract}
This paper considers 
program synthesis in the context of computational hardness, 
asking the 
question: How hard is it to determine whether a given 
synthesis problem has a solution or not?

To answer this question, this paper studies program synthesis for 
a basic imperative, Turing-complete language IMP, 
for which this paper proves that program synthesis is 
$\Sigma_3^0$-\emph{complete} in the arithmetical hierarchy.
The proof of this fact relies on a fully constructive encoding 
of program synthesis (which is typically formulated as a second-order 
query) as a first-order formula in the 
standard model of arithmetic 
(i.e., Peano arithmetic).
Constructing such a formula then allows us to reduce the 
decision problem for COF
(the set of functions which diverge only on a finite set of inputs), 
which is well-known to be a $\Sigma_3^0$-complete problem, 
into the constructed first-order representation of synthesis.

In addition to this main result, 
we also consider the 
hardness of variants of synthesis problems, such as those
introduced in previous work to make program synthesis 
more tractable (e.g., synthesis over finite examples).
To the best of our knowledge, this paper is the first to 
give a first-order characterization of program synthesis in general, 
and precisely define the computability of synthesis problems 
and their variants.

\end{abstract}

%% file: 1intro.tex
\section{Introduction}
\label{Se:intro}

In recent years, a vast amount of research has been conducted on the 
subject of \emph{program synthesis}, the task of automatically 
finding a program that meets a given logical specification.
Program synthesis is now 
finding a 
%used
wide variety of applications, 
such as in domain-specific languages (DSLs)~\cite{flashfillplus, swizzle}, 
invariant synthesis~\cite{freqplus}, 
or program repair~\cite{s3, angelix}, 
with a plethora of solvers~\cite{sketch, rosette, cvc4, semgus}  
%both general~\cite{sketch, rosette, cvc4, semgus} and specialized, 
capable of synthesizing programs for these applications.
%But another significant reason that synthesis remains such a popular field 
%of study is because program synthesis itself is a 
%theoretically interesting field of study.

Despite these advances, 
program synthesis remains a challenging topic in which we understand 
surprisingly little theoretically.
This is in part due to 
program synthesis being perceived as a 
computationally `hard' problem, 
albeit with good reason: 
because synthesis relies on being able to verify
a candidate program with respect to a given specification $\phi$, 
it is at the very least \emph{undecidable}.
Moreover, a synthesis problem is typically defined via 
a formula in the form of \eqref{synth-intro}, where 
$G$ denotes a grammar (that defines a search space of programs) and 
$D$ a domain of inputs:%, especially in the realm of computer science.
\begin{equation}
  \label{Eq:synth-intro}
\exists f, f \in G. \forall x, x \in D. \phi(f(x), x)
\end{equation}
Here,
the existential over the 
function $f$ makes \eqref{synth-intro} a formula in 
\emph{second-order logic}---the theory of which remains relatively 
less developed and understood compared to first-order logic.

Because synthesis in general is perceived to be
so difficult, much research in the area 
has focused on practical algorithms and heuristics targeted at solving  
synthesis problems with restrictions making them more tractable 
(e.g., synthesis over finite sets of examples~\cite{sketch}, or 
over limited DSLs~\cite{flashfillplus, swizzle}),
instead of the theory of synthesis itself.
%While it is true that 
%such focused approaches have led to many practical applications,
%it is also true that they unfortunately contribute 
%little to one of the reasons why previous work has 
%focused on solving tractable subproblems, instead of general solutions,
%in the first place: because we understand so little about 
%synthesis in general.

\myparvs{The Computational Hardness of Program Synthesis}
In this paper, we aim to shed some light on the theory of program 
synthesis, by studying program synthesis over a minimal, imperative, 
Turing-complete language
$\ltgt$, containing loops.
In particular, we ask the following fundamental question on the hardness of solving 
synthesis problems:

%answering a fundamental question on the hardness of solving 
%synthesis problems: 
%related to why research on program synthesis in general has been so limited: 
% IMP Talk

\begin{center}
	\emph{
	"How hard is it to determine whether
	a given synthesis problem has a solution or not?"}
\end{center}

%The answer to this question, as far as the authors are aware, 
%%has been mostly folklore: it is well-known that synthesis is 
%undecidable, but no existing work goes beyond this fact 
%and states how exactly hard solving a synthesis problem is.
%The answer to this question has been mostly folklore, in that 
%it is well-known that synthesis is undecidable, but none 
%has been able to give a concrete answer beyond this fact.
This paper gives a precise mathematical answer to this question:
program synthesis over Turing-complete languages 
%as defined in \eqref{synth-intro}
% encoded as a first-order formula within the standard model of arithmetic, which 
is $\sic{3}$-\emph{complete} in the 
arithmetical hierarchy.

Our proof of this fact relies on the fact that program synthesis can 
(somewhat surprisingly) be encoded as a \emph{first-order} formula within 
the standard model of arithmetic.
We prove this fact by giving a \emph{fully constructive} encoding of an 
arbitrary synthesis problem as a first-order formula.
While there exist other methods for showing that program synthesis has 
a first-order representation, the construction in this paper has an advantage 
in that it 
explicitly preserves the components of a synthesis problem, such as the 
grammar check $f \in G$, or the specification $\phi$.
This makes it much easier to treat the first-order representation itself as a 
synthesis problem as opposed to a complex blackbox formula, 
which in turn makes studying the properties of the first-order representation easier.

Based on the first-order representation of program synthesis, 
we then prove that program synthesis over $\ltgt$ is $\sic{3}$-complete,
by reducing a well-known $\sic{3}$-complete problem into the 
constructed first-order representation
(something that would not have been possible using only the second-order definition).
The $\sic{3}$-complete problem in question is \cof~\cite{soare}, the set of 
all functions which diverge only on a finite set of inputs---we show that 
membership in \cof \xspace can be reduced to determining whether a synthesis problem 
has a solution or not.

The fact that our first-order encoding takes care to preserve components of the 
original synthesis problem 
also allows us to easily consider \emph{variants} 
of these components, 
such as restrictions from previous work
to make synthesis more tractable.
For example, such a common restriction is \emph{programming-by-example}, where 
the input domain consists of a finite number 
of examples~\cite{sketch,flashfill};
our construction immediately yields the fact that programming-by-example is 
$\sic{1}$-complete, whereas, 
to the best of our knowledge, there has been 
no exact analysis of how much limiting the input domain makes a 
synthesis problem easier to solve.
Like this, 
the results in this paper help precisely quantify how much of an effect 
certain restrictions have on a synthesis problem, 
shedding light onto the general observations from previous work that these 
restricted synthesis problems are practically easier to solve.

%In addition to the surprising fact that program synthesis is first-order, 
%bringing synthesis into the realm of first-order logic via such an explicit construction 
%now allows us to study program synthesis as a first-order formula, 
%which is much better understood compared to the original second-order 
%formulation.

%Most importantly, 

% To look at the two contributions in more detail, the first contribution---that 
% program synthesis is first-order---takes the traditional second-order logic 
% formula that defines program synthesis, and shows that there exists an \emph{equivalent} 
% first-order formula with the same semantics.
% In particular, the proof \emph{explicitly constructs} the equivalent first-order query, 
% in a manner such that the elements of the original and the reduced query can be matched;

At this point, one may ask: what practical merit is there in studying 
the hardness of program synthesis, seeing that almost everything is undecidable 
anyways?
While it is true that the results in this paper are mostly theoretical in this sense, 
we argue that these theoretical results actually suggest future research directions 
that may lead to more practical synthesis algorithms in the future.
For example, when designing synthesis algorithms for a specific application 
(e.g., data movement expressions for GPUs~\cite{swizzle}), 
one may now derive exactly how hard the problem at hand is, and also 
consider realistic restrictions for the application that would make 
the problem easier 
(e.g., the movement expression only needs to be correct on a finite range of data).
Such insight is often valuable in the design of specialized algorithms.

Another potentially interesting notion that our paper unearths is that of 
\emph{generalization}, the task of extending a function 
that is correct on a finite subset of inputs (e.g. one that is obtained 
via programming-by-example)
to be correct on a general, infinite set of inputs.
Many solvers rely on some form of generalization~\cite{sketch, semgus}, but 
the task of generalization itself is one that has received nearly no
previous research---perhaps because there exists a well-known generalization 
procedure, \emph{counterexample-guided inductive synthesis}~\cite{sketch} 
(CEGIS), that works very well in practice.
However, in \S\ref{Se:variants}, we show that CEGIS is actually suboptimal 
as a generalization method, at least
in terms of computational hardness: 
hinting at possibly more efficient methods for generalization.
In tandem with identifying the generalization problem itself, we argue 
that the results in this paper, while themselves theoretical, 
provoke such interesting questions for future work.

\myparvs{Contributions}
To summarize, we make the following contributions:
\begin{itemize}
  \item A fully constructive reduction of the standard 
    second-order formula that defines program 
    synthesis into an equivalent first-order formula, showing 
    that program synthesis is \emph{first-order} (\S\ref{Se:fo}).
  \item A classification of the first-order synthesis query developed in \S\ref{Se:fo} 
    with respect to the arithmetical hierarchy, showing that program synthesis is 
    $\sic{3}$-\emph{complete} (\S\ref{Se:completeness}).
  \item Based on the results from \S\ref{Se:fo} and \S\ref{Se:completeness}, 
    an analysis on the computability of \emph{variants} 
    of synthesis problems, such as those introduced in previous 
    work to make synthesis more 
    tractable such as programming-by-example (\S\ref{Se:variants}).
\end{itemize}

\S\ref{Se:prelim} discusses preliminary concepts in computability, 
and defines the target language we will use in this paper.
\S\ref{Se:discussion} discusses the results and contributions of this paper.
\S\ref{Se:conclusion} concludes.

%% file: 2prelim.tex
\section{Preliminary Concepts and the Language Definition}
\label{Se:prelim}

In this section, we introduce necessary preliminary definitions and theorems from 
computability theory, and define the target language we will consider for 
defining program synthesis problems.

\subsection{Preliminary Concepts from Computability}
As stated in \S\ref{Se:intro}, the goal of this paper is to answer 
how hard it is to solve a program synthesis problem in a precise mathematical manner.
We rely on concepts from computability 
theory to answer this question, 
%formally define the concept of computational hardness, 
starting with the definition of a \emph{problem}.

\begin{definition}[Problem]
  \label{def:problem}
  A \emph{problem} $P$ is defined as a subset of the natural numbers $\mathbb{N}$.
  We define the \emph{decision problem} for $P$ as to determine whether 
  a given arbitrary number $x$ is a member of $P$.

  If there exists an algorithm capable of solving the decision problem for $P$ within 
  a finite amount of steps, we say that $P$ is \emph{decidable}, \emph{computable}, or 
  \emph{recursive}.
  If not, we say that $P$ is \emph{undecidable}.

  If there exists an algorithm capable of solving the decision problem for $P$
  when $x \in P$ in a finite amount of steps, 
  but may not terminate otherwise, we say that $P$ is 
  \emph{recursively enumerable}, or \emph{semidecidable}.
  Similarly, if there exists an algorithm capable of solving cases where $x \not \in P$, 
  but may not terminate otherwise, we say that $P$ is 
  \emph{co-recursively enumerable}.
\end{definition}

Definition~\ref{def:problem} defines a problem as the set of 
its solutions, which are encoded as natural numbers.
For example, the canonical statement of the Halting problem
is the set $\set{(M, i) \mid M \text{ halts on input } i}$, where the Turing machine-input 
pair $(M, i)$ can be further encoded as a single number.
The fact that problems can be defined as sets of natural numbers is important, 
as it allows us 
to express problems using a seemingly completely different formalism: 
formulas in the standard model of arithmetic.

By the standard model of arithmetic, we refer to the first-order theory of 
Peano arithmetic, in which one can write 
formulae that contain addition, multiplication, and first-order quantifiers 
ranging over natural numbers 
(which can be further used to encode other operations such as division or remainder).
For this paper, we will assume that the standard model also contains 
a symbol for each primitive recursive function; i.e., that the standard model of 
arithmetic has been extended with the axioms for primitive recursion.
This makes studying formulae in the standard model more natural from the perspective 
of computability, as we will shortly see.
For the remainder of this paper, we will use the word `formula' to refer to such 
first-order formulae, unless otherwise explicitly noted.

The set of numbers that a formula $\phi$ defines are exactly the set of numbers 
(i.e., assignments) that make 
$\phi$ true when substituted for the free variables: for example, the formula 
$\phi(x) = \exists y. x + y = 2$ defines the set of all $x$ for which there exists a $y$ 
such that $x + y = 2$ (namely, the set $\set{0, 1, 2}$).
It is now clear that sets, problems, and formulae are 
simply different ways of expressing the same object, and we will use these 
terms interchangeably in this paper from this point on.

One key idea that we will rely on for this paper is that it is possible 
to encode 
\emph{any arbitrary sequence of finite length} as a \emph{pair of numbers} 
in the standard model of arithmetic, 
a construction that will be essential in our 
reduction of program synthesis to a first-order formula in \S\ref{Se:fo}.
The construction relies on the \godel $\beta$-function, 
which defines a way to decode a pair of integers into 
a finite sequence of integers.
\begin{lemma}[The \godel $\beta$-function~\cite{godel}]
  \label{lem:beta}
  Let $\beta(a, b, i)$ define the function $\beta(a, b, i) \equiv a (\mymod 1 + b \cdot (1 + i))$.
  Then for any finite length $l$ and a sequence of integers $\lrangle{c_0, \cdots, c_l}$ of length $l$, 
  there exists integers $n_a, n_b$ such that the following holds:
  \[
     \forall j, 0 \leq j \leq l. \beta(n_a, n_b, i) = c_i
  \]
\end{lemma}
The proof of Lemma~\ref{lem:beta} constructs the values $m$ and $n$ through the 
Chinese remainder theorem; the actual construction is not important for this paper, 
and we refer the reader to~\citet{godel} for details.
The important part of Lemma~\ref{lem:beta} is that sequences of \emph{unbounded but finite} 
length may be encoded as a pair of integers.
In later sections, we will rely 
on this fact to encode information such as the syntactic structure of programs, 
which may be encoded as a sequence of productions, 
as first-order formulae.

\myparvs{The Arithmetical Hierarchy}
Intuitively, the more quantifiers a formula has, the harder it will be to check 
whether an assignment satisfies the formula or not: an increase in the 
hardness of computability of that set.
This allows us to study the difficulty of solving a problem by studying 
its representation as a formula, an 
intuition which is formalized via the arithmetical hierarchy.

%\footnote{
%  Intuitively speaking, the reason why \emph{alternating} quantifiers have an effect on the computability 
%  of a problem, while non-alternating quantifiers do not, is because a finite number of consecutive non-alternating 
%  quantifiers can always be compressed into a single quantifier that encodes the same information.
%}

\begin{definition}[The Arithmetical Hierarchy]
  \label{def:ah}
  Let $\phi$ be a first-order formula in the standard model of arithmetic, in prenex normal form.
  The \emph{arithmetical hierarchy} consists of two sequences of classes, 
  $\sic{n}$ and $\pic{n}$, where $n$ is a natural number.
  $\phi$ is assigned a class in the arithmetical hierarchy as following:
  \begin{itemize}
    \item If $\phi$ contains only bounded quantifiers, then $\phi$ is both $\sic{0}$ and $\pic{0}$; these two classes are equivalent.
    \item If $\phi$ is of the form $\exists x. \psi$ for a variable $x$ and a formula $\psi$
      with classification $\pic{n}$ (that is, $\psi$ contains $n$ alternating unbounded quantifiers led by a $\forall$), 
      then $\phi$ is in the class $\sic{n + 1}$. 
    \item If $\phi$ is of the form $\forall x. \psi$ for a variable $x$ and a formula $\psi$
      with classification $\sic{n}$ (that is, $\psi$ contains $n$ alternating unbounded quantifiers led by an $\exists$), 
      then $\phi$ is in the class $\pic{n + 1}$. 
  \end{itemize}
\end{definition}

Because it is always possible to add redundant quantifiers to a formula without altering its meaning
(e.g., $\forall z. \exists y. x + y = 2$ denotes the same set of numbers as $\exists y. x + y = 2$), 
a formula $\phi$ in $\sic{n}$ or $\pic{n}$ is guaranteed to be in $\sic{k}$ and $\pic{k}$ for all 
$k > n$ as well.
%In this sense, the arithmetical hierarchy creates a hierarchy of problems, 
%where problems in the upper classes are harder compared to those in the lower classes.
Thus when studying the computability of a particular problem $P$, one will be interested in the 
finding \emph{lowest class} that a problem can be assigned---in other words, 
constructing a formula with the 
\emph{minimum} amount of quantifiers that expresses the same set as $P$.

To bring the discussion back to computability, 
Theorem~\ref{thm:ah} connects classes from the arithmetical hierarchy to the concepts 
of computability defined in Definition~\ref{def:problem}.
\begin{theorem}%[Computability of the Arithmetical Hierarchy]
  \label{thm:ah}
  Let $P$ be a problem. Then the following holds:
  \begin{itemize}
    \item P is decidable if and only if $P \in \sic{1} \cap \pic{1}$.
    \item P is r.e. if and only if $P \in \sic{1}$.
    \item P is co-r.e. if and only if $P \in \pic{1}$.
  \end{itemize}
  Furthermore, let an \emph{oracle} for a class $\sic{n}$ or $\pic{n}$ be a mechanism that can 
  instantly solve the decision problem for all problems in that class.
  Then the following holds:
  \begin{itemize}
    \item If $P$ is in $\sic{n+1}$, then $P$ is r.e. given access to an oracle for $\pic{n}$.
      That is, there exists an algorithm capable of answering $\Etrue$ for cases in which 
      $x \in P$ within a finite number of steps, if the algorithm can access the oracle for $\pic{n}$.
    \item If $P$ is in $\pic{n+1}$, then $P$ is co-r.e. given access to an oracle for 
      $\sic{n}$.
  \end{itemize}
\end{theorem}
Perhaps most interesting are the latter two bullets, which formalize 
how much harder problems higher up in the hierarchy is compared to one in a lower class: 
they become semidecidable when given an oracle for problems directly one class lower.

The way we defined the class $\sic{0} = \pic{0}$, along with Theorem~\ref{thm:ah}, explain 
why we allow symbols for primitive recursive functions in our formulae: 
recursion for these formulae are guaranteed to terminate in a finite number 
of steps (similar to bounded quantifiers).
However, expressing primitive recursive functions in a language without
symbols for primitive recursion functions 
requires the use of an existential quantifier, 
which changes their classification in the arithmetical hierarchy.
Thus, when studying these formulae for the purpose of computability, it is beneficial 
to allow primitive recursion in $\sic{0} = \pic{0}$.
We note that starting from $\sic{1}$ and $\pic{1}$, classes higher up in the hierarchy 
remain unchanged regardless of whether we allow primitive recursive operators 
in our formulae or not.

%In Definition~\ref{def:ah}, note how \emph{bounded} quantifiers 
%do not affect the computability of a formula.
%Intuitively, this is because bounded quantifiers may be expanded 
%a finite amount of times to remove the quantifier; we take note of this fact as bounded quantifiers 
%will appear often in \S\ref{Se:fo} and \S\ref{Se:completeness}.

We conclude this section with a brief recap on the classifications of well-known problems in 
computer science, which will be useful later when stating the hardness of program synthesis 
relative to these problems.
The Halting problem is in $\sic{1}$; furthermore, it is $\sic{1}$-\emph{complete}, meaning 
that any other problem in $\sic{1}$ can be reduced to a version of the Halting problem.
The universal Halting problem, which asks for the set of Turing machines 
$\set{M \mid M \text{ halts on every input}}$, is $\pic{2}$-complete.

In program verification and synthesis, most of the specifications that a program 
is desired to meet are in $\sic{0} = \pic{0}$ 
(the primitive recursive specifications).
Based on this notion, standard safety verification, where the goal is to 
show that a program terminates while satisfying some safety property on 
all inputs, is also $\pic{2}$-complete.
The result of this paper prove that program synthesis is $\sic{3}$-complete, which 
means that, if one is given an oracle for safety verification,
program synthesis becomes semidecidable.

\subsection{Defining the Target Language IMP}

Having established the necessary concepts related to computability, we now 
define the target language $\ltgt$ we will use for defining synthesis 
problems.
Figure~\ref{fig:gtgt} defines the grammar $\gtgt$ that generates $\ltgt$, which is 
a minimal but still Turing-complete 
imperative language that contains variables, Boolean and integer expressions,
assignments, sequential composition, branches and loops.
Note that the maximum arity of operators in $\gtgt$ is $2$ 
(e.g., $\gtgt$ contains only $\mathsf{if\ then}$ as opposed to $\mathsf{if\ then\ else}$): this 
will become useful in \S\ref{Se:fo} when constructing a first-order formula 
for program synthesis.

\begin{figure}
   \[
  \begin{array}{llcl}
    \mathit{Boolean}    & B & ::= &  \Etrue \mid \Efalse \mid ! B \mid B \wedge B \mid E < E \mid E = E \\
    \mathit{Variable}   & V & ::= &  x \mid y \mid \cdots \\
    \mathit{Expression} & E & ::= &  0 \mid 1 \mid V \mid E + E \mid E - E \mid E \cdot E \mid E / E \\
    \mathit{Statement}  & S & ::= &  \Eassign{V}{E} \mid \Eseq{S}{S} \mid \Eif{B}{S} \mid \Ewhile{B}{S} \\
  \end{array}
  \]
  \caption{The target grammar $\gtgt$ that generates the target language $\ltgt$ we are interested in for this paper. 
  }
  \label{fig:gtgt}
\end{figure}

We define a semantics for terms inside $\ltgt$.
In this paper, we define a state $\sigma$ as a map $\mathsf{Variable} \rightarrow \mathsf{Value}$,
where a $\mathsf{Value}$ is an integer.
Such a state $\sigma$, and also updates to the state,
may be further encoded within the standard model using tuples of values.
The semantics of an arbitrary term
$t \in \ltgt$, denoted as $\sem{t}$, is understood to be a (partial) function 
which takes as input a state and 
either produces a new state (for statements), or
integer or Boolean values (for expressions, variables and Booleans).

There are many ways to define semantics for terms, such as big-step semantics; 
\eqref{while-sem-recursive} gives an example of defining semantics for loops 
in this fashion.
{\small
\begin{equation}
  \label{Eq:while-sem-recursive}
  \sem{\Ewhile{b}{s}}(\sigma) = \Mifthenelse{\sem{b}}{\sem{\Ewhile{b}{s}}(\sem{s}(\sigma))}{\sigma}
\end{equation}
}
In this paper, we focus on the fact that the semantics of a term $t$
can also be represented as a formula, as illustrated in Lemma~\ref{lem:while-sem}.
\begin{lemma}[\citet{winskel}]
  \label{lem:while-sem}
  Consider a loop $t = \Ewhile{b}{s}$ for $b \in L(B)$ and $s \in L(S)$.
  Then the result of executing $t$ on an input state $\sigma$ is equivalent to the set of $\pi$ that 
  satisfy the following formula:
  {\small
  \begin{equation}
    \label{Eq:while-sem}
  \begin{split}
    \sem{\Ewhile{b}{s}}(\sigma) = & \exists k. \exists \sigma_0, \cdots, \sigma_{k}. (\sigma = \sigma_0) \wedge \\
      & (\forall i, 0 \leq i < k. \sem{b}(\sigma_{i}) = \Etrue \wedge \sem{s}(\sigma_{i}) = \sigma_{i + 1}) \wedge \\
      & (\pi = \sigma_k \wedge \sem{b}(\sigma_{k}) = \Efalse)
  \end{split}
  \end{equation}
  }
  In other words, \eqref{while-sem} captures the semantics of $t$ as a formula.
\end{lemma}
The two discussed ways for defining the semantics of loops 
(\eqref{while-sem-recursive} and \eqref{while-sem}) are equivalent: it is possible to prove that 
\eqref{while-sem} is a fixed-point solution of \eqref{while-sem-recursive} (\citet{winskel}).
It is thus once again beneficial, from the viewpoint of studying computability, to 
express semantics as formulae; this will be useful in our encoding of program 
synthesis as a first-order formula in \S\ref{Se:fo}.

We also observe that strictly speaking, 
\eqref{while-sem} itself is not a first-order formula: it contains a quantified sequence of variables 
$\exists \sigma_0, \cdots, \sigma_{k}$ of variable length $k + 1$.
However, because $k$ is finite, Lemma~\ref{lem:beta} guarantees the existence of a pair of numbers 
$a_{\sigma}$ and $b_{\sigma}$ that encode $\sigma_0, \cdots, \sigma_k$ through the \godel $\beta$-function; 
this allows one to replace the sequence $\sigma_0, \cdots, \sigma_k$ with $a_{\sigma}$ and $b_{\sigma}$ 
to obtain a true first-order formula.

We wrap the discussion on $\ltgt$ up by 
considering the  structure of \eqref{while-sem}, which can be understood intuitively 
as a ``guess-and-check'' system.
The idea is that 
the existential quantifier over $\sigma_0, \cdots, \sigma_k$ ``guesses'' the correct sequence of states that 
would be obtained from iterating through the while loop, while the ensuing formula ``checks'' that 
the guessed sequence actually respects the semantics of the loop.
We will rely extensively on similar structures, where we guess the correct sequence of states 
then validate them, in \S\ref{Se:fo}.

%% file: 4fo.tex
\section{Constructing a First-Order Representation of Program Synthesis}
\label{Se:fo}

Having established the necessary preliminaries, this section sets the stage 
for studying synthesis from the perspective of computability by 
constructing a first-order formula that is equivalent to program synthesis.
There do exist alternative ways to show that program synthesis is first-order, 
but as discussed in \S\ref{Se:intro}, the construction in this paper is unique 
in that it provides an explicit construction of each component in the definition 
of a synthesis problem as a first-order formula.
We will see later in \S\ref{Se:variants} and \S\ref{Se:discussion} that such 
an explicit, structure-preserving construction is highly beneficial in 
studying the complexity of synthesis problems.

Program synthesis problems are often defined using two components: 
\rone a \emph{grammar} $G$, that defines a possibly infinite set of terms, and 
\rtwo a \emph{specification} $\phi$ , which is a first-order Boolean formula
specifying the behavior of the function to be synthesized.
For the purposes of this paper, we will assume that the grammar $G$ is given as a 
regular tree grammar (RTG)
such that $L(G) \subseteq L(\gtgt)$; i.e., 
$G$ is only allowed operators within $\gtgt$ as part of its productions.
When we refer to an arbitrary grammar in the rest of this paper, we refer to such 
a subgrammar of $\gtgt$.
%\begin{example}[Subgrammars of $\gtgt$]
%  \label{ex:subgrammar}
%  Consider the following two grammars:
%
%  \begin{minipage}{0.48\textwidth}
%    \begin{equation}
%      \label{Eq:grm_2}
%      \begin{array}{lcl}
%        \mathit{B} & ::= & E < E \mid ! B \\
%        \mathit{E} & ::= & x \mid y \mid E + E \\
%        \mathit{S} & ::= & \Eassign{x}{E} \mid \Eassign{y}{E} \mid \Ewhile{B}{S} \\
%        \mathit{L} & ::= & S \mid \Eseq{L}{L} \mid \Ewhile{B}{L}
%      \end{array}
%    \end{equation}
%  \end{minipage}
%   \begin{minipage}{0.48\textwidth}
%    \begin{equation}
%      \label{Eq:grm_3}
%      \begin{array}{lcl}
%        \mathit{B} & ::= & E < E \mid B \vee B \\
%        \mathit{E} & ::= & x \mid y \mid E + E \mid E - E \\
%        \mathit{S} & ::= & \Eassign{x}{E} \mid \Eifthenelse{B}{S}{S}
%      \end{array}
%    \end{equation}
%  \end{minipage}
%
%  The grammar~(\ref{Eq:grm_2}) is a subgrammar of $\gtgt$, as all terms that may be generated 
%  by the grammar also have a valid derivation tree in $\gtgt$.
%  On the other hand, the grammar~(\ref{Eq:grm_3}) is not a subgrammar of $\gtgt$, as 
%  it contains operators such as $\vee$ and $\mathsf{if}\ \mathsf{then}\ \mathsf{else}$ that are not in 
%  $\gtgt$.
%\end{example}
%As discussed in \S\ref{Se:lang}, this restriction can be lifted by extending $\gtgt$ to 
%contain more operators.

Similarly to the restriction on grammars, we will also assume for the time being 
that the specification $\phi$ is in the class $\sic{0} = \pic{0}$ as discussed in \S\ref{Se:prelim};
when we refer to an arbitrary specification in this section, we refer to such a 
primitive recursive formula.
%As observed in \S\ref{Se:prelim}, this assumption is well-founded in the sense that 
%it consists the set of specifications that are checkable for at least single inputs; 
Later in \S\ref{Se:variants}, we will consider specifications that are outside of this 
assumption and the effect they have on the hardness of synthesis problems.

\begin{definition}[Synthesis Problem]
  \label{def:synth-def}
  Let $G$ be a RTG, $D$ be a domain of input states, and
  $\phi$ a specification.
  We define a \emph{synthesis problem} $\sy$ over $G, D, $ and $\phi$ 
  as the following second-order formula:\footnote{
    We include the term $f$ as an argument to $\phi$ here (unlike \eqref{synth-intro})
    as some synthesis problems 
    place restrictions on the structure of $f$ itself, such as the size of $f$.
  }
  \begin{equation}
  \label{Eq:synth-def}
  \sy \defeq \exists f, f \in G. \forall \sigma, \sigma \in D. 
    \phi(\sigma, f, \sem{f}(\sigma))
  \end{equation}

  If \eqref{synth-def} is true, meaning that there does exist an $f \in G$ that meets the 
  specifications of $\sy$, we say that $\sy$ is \emph{realizable}.
  Otherwise, we say that $\sy$ is \emph{unrealizable}.
\end{definition}

\eqref{synth-def} is traditionally understood as a second-order formula because 
it contains an existential quantifier over a term $f$---essentially a relation that 
relates an input state and an output state, where the two 
components concerning $f$ in \eqref{synth-def} are $f \in G$ and $\sem{f}(\sigma)$.
Our goal in this section is to replace these two components with equivalent 
first-order representations.

\myparvs{Key Idea}
To construct first-order formulae encoding to these components,
consider $f$ in \eqref{synth-def} as a purely 
syntactic first-order object instead of a function 
(i.e., a natural number that encodes the syntactic structure of $f$).
The first component concerning $f$ in \eqref{synth-def}, $f \in G$, 
is a syntactic check over the structure of $f$, which needs not consider the fact 
that $f$ is a second-order relation.
In the latter component, $\sem{f}(\sigma)$,
it is also possible to treat $f$ simply as an argument passed to the semantics 
function $\sem{\cdot}$, if one takes the perspective that $\sem{\cdot}$ 
itself is a function of type 
$\mathsf{Term} \rightarrow \mathsf{State} \rightarrow \mathsf{State}$.

The core challenge then becomes constructing the semantics function $\sem{\cdot}$ 
as a first-order formula.
Observe that unlike Lemma~\ref{lem:while-sem}, for $\sem{\cdot}$ we must 
construct a formula that is capable of accepting and analyzing the 
semantics of \emph{every} term $t \in L(G)$, while rejecting terms that are not 
in $G$, for an arbitrary grammar $G$.
The main contribution of this section are the key constructions required for 
constructing the formula for $\sem{\cdot}$, 
which provide enough insight on the 
first-order representation of synthesis to derive our results on computational hardness 
in \S\ref{Se:completeness} and $\S\ref{Se:variants}$.

\input{4-1syntax.tex}

\input{4-2semantics.tex}

\input{4-3loops.tex}

%% file: 4-1syntax.tex
\subsection{Dealing with Syntax: Complete Binary Trees}
\label{SubSe:syntax}

To construct the first-order representation for 
$\sem{\cdot}$, one must first be able to parse terms correctly using 
first-order formulae in order to encode their semantics and check 
whether they are syntactically valid.
Because $f$ is a term which is syntactically represented as a tree 
in \eqref{synth-def}, 
but we wish to treat $f$ simply as a natural 
number for the first-order construction, 
our main goal becomes encoding arbitrary trees as natural numbers.
There are many ways to perform this encoding; in this section, we will 
develop our own variant that is particularly well-suited to the task 
of constructing a corresponding formula for $\sem{\cdot}$.

To encode trees as natural numbers, we will rely on the fact that 
\rone a tree $f$ with $l$ nodes 
can be represented as a sequence $p_0, \cdots, p_l$
(e.g., by taking a preorder traversal of $f$) , and 
\rtwo such a finite sequence may be encoded via a 
pair of integers $a_p, b_p$ following Lemma~\ref{lem:beta}.
The main challenge lies in the fact that 
preorder traversals of trees are not guaranteed to be unique---for 
example, Diagram~\ref{Eq:diag_1} illustrates two distinct terms which 
result in the same preorder traversal, but only the left term
is syntactically valid.
\begin{minipage}{0.9\linewidth}
\begin{center}
  $\quad$
  \begin{tikzpicture}[scale=.9]
    \Tree [.== [.+ 1 2 ] 3 ]
  \end{tikzpicture}
  \hspace{3mm}
  \begin{tikzpicture}[scale=.9]
    \Tree [.== + [.1 2 3 ] ]
  \end{tikzpicture}
\end{center}
\end{minipage}
\begin{minipage}{0.05\linewidth}
\begin{equation}
  \label{Eq:diag_1}
\end{equation}
\end{minipage}

\myparvs{Fixing Structure via Complete Binary Trees}
To fix this problem, in this paper, 
we will fix the structure of the trees that we are considering to 
\emph{complete binary trees}, which
%The idea is simple---complete binary trees 
have a fixed structure in which 
a parent node with index $i$ has exactly two children located at $2i + 1$ and 
$2i + 2$, assuming a left-to-right preorder traversal with root node $0$.
This in turn results in a unique sequence generated via preorder traversal 
for each tree.

The problem with fixing the structure of trees to complete binary trees is that 
not every term $f$ in a grammar $G$ is guaranteed to be complete binary.
However, we observe the fact that such $f$ can still be 
\emph{embedded} within a complete binary tree, 
because the maximum arity of operators in $G$ is two (following the definition 
of $\gtgt$ in \S\ref{Se:prelim}).
Based on this observation, we will construct an extended grammar $G_{\bin}$ 
for $G$, that intuitively adds some `dummy' syntax to $G$ with the 
following two goals:
\rone that every valid 
term inside $G$ has a corresponding complete binary representation 
in $G_{\bin}$, and 
\rtwo a synthesis problem $\sy$ defined over $G$ is realizable if and 
only if $\sy$ defined over $G_{\bin}$ (with the same domain and specification) 
is realizable.
This idea is formalized by Definition~\ref{def:complete-binary-form}, which defines 
the dummy syntax, and Definition~\ref{def:null-sem}, which defines its semantics.

%Note that accepting only complete binary trees is beyond the expressiveness 
%of a regular tree grammar, thus instead of manipulating $G$ into a grammar that
%only accepts complete binary trees, we will define a separate grammar 
%$G_{\bin}$ where $\sy$ (defined over $G$) is realizable if 
%and only if $\sy$ is realizable when defined over $G_{\bin}$ as well.

\begin{definition}[Complete Binary Form]
  \label{def:complete-binary-form}
  Let $G$ be some arbitrary subgrammar of $\gtgt$.
  Then one can define a new grammar $G_{\bin}$ as following, by adding a new 
  nonterminal $\nullnt$, 
  a binary operator $\nullop$, and a $0$-ary leaf operand $\nullleaf$:
  \begin{itemize}
    \item All operators and operands originally in $G$ are binary in $G_{\bin}$ 
      (e.g., $0$ is a binary operator in $G_{\bin}$).
    \item Operators whose arity has been increased via the previous item 
      have may only have their additional operands as the new nonterminal $\nullnt$ 
      (e.g., $0$ is now $0(\nullnt, \nullnt)$).
    \item $\nullnt$ consists of the productions \\
      $\nullnt ::= \nullleaf \mid \nullop(\nullnt, \nullnt)$.
  \end{itemize}
  Given a grammar $G_{\bin}$ defined in this manner, we say that 
  $G_{\bin}$ is the \emph{complete binary form} of $G$.
\end{definition}

$\nullnt, \nullop,$ and $\nullleaf$ consist the dummy syntax; a complete binary tree 
in $G_{\bin}$ can be treated as a tree in which a term from $G$ is embedded at the root.

Following Definition~\ref{def:complete-binary-form}, 
we must give a semantics to the dummy syntax such that, as previously mentioned, 
a synthesis problem $\sy$ defined over a grammar $G$ is realizable if and only if
$\sy$ defined over the complete binary form of $G$ is realizable.
Intuitively, this can be achieved by letting the dummy nodes result in dummy values, 
while the non-dummy nodes retain their original semantics.

\begin{definition}[Semantics of $\nullop$ and $\nullleaf$]
  \label{def:null-sem}
  Let $\nullval$ denote a dummy value.
  Then the semantics of $\nullop$ is defined as $\sem{\nullop(t_1, t_2)}(\sigma) = \nullval$ for 
  arbitrary terms $t_1, t_2 \in \ltgt$ and an arbitrary state $\sigma$.

  The semantics of operators originally in $\ltgt$ are modified such that if 
  at least one of their operands are $\nullval$, then the operator also yields $\nullval$.
  For example, the semantics of $+$ may be defined as:
  {\small
  \[
    \sem{t_1 + t_2}(\sigma) = 
      \Eifthenelse
        {(\sem{t_1}(\sigma) = \nullval \vee \sem{t_2}(\sigma) = \nullval)}
        {\nullval}
        {\sem{t_1}(\sigma) + \sem{t_2}(\sigma)}
  \]
  }
  In particular, the semantics of operators whose arity are changed in the complete binary form 
  (e.g., the 0-ary variable $x$) remain unchanged.
  For example, $x$ is also now a binary operator, but the semantics of 
  $x$ remain unchanged, as $\sem{x}$ does \emph{not} depend on the behavior of its subterms:
  {\small
  \[
    \sem{x(t_1, t_2)}(\sigma) = \Eifthenelse{(\sigma = \nullval)}{\nullval}{\sigma[x]}
  \]
  }
\end{definition}

Definition~\ref{def:null-sem} formalizes the dummy semantics: dummy nodes result in
dummy values, and dummy values $\nullval$ are propagated if they appear as an operand.
An important part to note about this propagation is that in Definition~\ref{def:null-sem},
the semantics of $0$-ary and $1$-ary operators do not change: 
this prevents the final result of evaluating any term 
that contains $\nullleaf$ or $\nullop$ from being $\nullval$.

\begin{example}
  \label{ex:complete-binary-tree}
  Consider the following simple grammar $E$: 
  \[
    E ::= 1 \mid x \mid E + E
  \]
  The complete binary form $E_{\bin}$ of $E$ is defined as following:
  \begin{align*}
    E ::=       & \ 1(\nullnt, \nullnt) \mid x(\nullnt, \nullnt) \mid E + E \\
    \nullnt ::= & \ \nullleaf \mid \nullop(\nullnt, \nullnt)
  \end{align*}

  Consider a term $1 + x + 1$ (which can also be written as $+ (+ (1, x), 1)$ in preorder form) 
  in $L(E)$.
  This term has a corresponding complete binary tree representation in $E_{\bin}$, namely 
  $+ (+ (1 (\nullleaf, \nullleaf), x(\nullleaf, \nullleaf)),  
  1(\nullop(\nullleaf, \nullleaf), \nullop(\nullleaf, \nullleaf)))$.
 
  %\begin{figure}
  %  \begin{tabular}{cc}
  %    \begin{tikzpicture}[scale=.75]
  %     \Tree[.+ [.+ 1 $x$ ] 1 ]
  %   \end{tikzpicture}
  %    &
  %    \begin{tikzpicture}[scale=.75]
  %     \Tree[.+ [.+ [.1 $\nullleaf$ $\nullleaf$ ] [.$x$ $\nullleaf$ $\nullleaf$ ] ]
  %              [.1 [.$\nullop$ $\nullleaf$ $\nullleaf$ ] [.$\nullop$ $\nullleaf$ $\nullleaf$ ] ] ]
  %   \end{tikzpicture}
  %  \end{tabular}
  %  \caption{The term $1 + x + 1$ expanded as a tree, where the left tree is in $L(E)$ and the right 
  %  tree is in $L(E_{\bin})$.}
  %  \label{fig:tree-bad-figure}
  %\end{figure}
  
  %Observe that the first tree has an embedding in the second tree, and by virtue of the 
  %second tree being completely binary, one is capable of indexing all children nodes 
  %from a parent node.
  
  Observe that the semantics of the original and complete binary 
  terms are equivalent: this is because although 
  $\nullleaf$ and $\nullop$ result in dummy values, $0$-ary operators such as $1$ and $x$ do not 
  propagate these dummy values when computing their semantics.
  %consult these dummy values when computing the semantics of 
  % $+ (+ (1 (\nullleaf, \nullleaf), x(\nullleaf, \nullleaf)),  
  % 1(\nullop(\nullleaf, \nullleaf), \nullop(\nullleaf, \nullleaf)))$.
  %Thus the dummy nodes contribute nothing to the semantics of the complete binary term, and the part 
  %of the complete binary term that is relevant to the semantics is the part which embeds the term 
  %$1 + x + 1$---and thus the semantics of the two terms are equivalent.
\end{example}

Definition~\ref{def:null-sem} allows us to state the correctness of 
the complete binary form as a theorem.

\begin{theorem}[Soundness of the Complete Binary Form]
  \label{thm:bin-soundness}
  Let $\sy$ be a synthesis problem defined over a grammar $G$ and specification $\phi$, 
  and $G_{\bin}$ be the complete binary form of $G$.
  Then there exists $f \in G$ such that $\sy$ is realizable, if and only if there exists
  a complete binary 
  $f_{\bin} \in G_{\bin}$ such that $\sy$ defined over $G_{\bin}$ and $\phi$ is realizable.
\end{theorem}

\begin{proof}
  To see that $\sy$ defined over $G_{\bin}$ is realizable if $\sy$ defined over $G$ is realizable,
  we will show that for every term $f \in G$, there exists a corresponding $f_{\bin} \in G_{\bin}$
  such that $f_{\bin}$ is a complete binary tree and $\sem{f} = \sem{f_{\bin}}$.
  The proof for this is simple: let the height (of the tree representation of) $f$ be $h$.
  Create a complete binary tree $f_{\bin}$ of height $h + 1$, such that $f$ is embedded into $f_{\bin}$ such
  that their root nodes coincide, and the rest of the nodes are dummy nodes ($\nullop$ if the node has children,
  $\nullleaf$ if it does not).
  $f_{\bin}$ is guaranteed by construction to be a term of $G_{\bin}$ and to have a semantics equivalent to
  $\sem{f}$, and thus $\sy$ defined over $G_{\bin}$ is realizable through $f_{\bin}$.

  To see that $\sy$ defined over $G$ is realizable if $\sy$ defined over $G_{\bin}$ is realizable,
  observe that if $\sy$ over $G_{\bin}$ is realizable, there exists $f_{\bin} \in G_{\bin}$
  such that $\sem{f_{\bin}} \neq \nullval$.
	Remove all dummy nodes in $f_{\bin}$ to get $f \in G$ that witnesses
  the realizability of $\sy$ defined over $G$.	
\end{proof}

Notice a small caveat with Theorem~\ref{thm:bin-soundness}, in that it says nothing about
$f_{\bin}$ being a complete binary tree.
This is because accepting only complete binary trees is outside the power of a regular grammar.
Nevertheless, the construction of $G_{\bin}$ guarantees
that every term in $G_{\bin}$ also has an equivalent complete binary representation,
which ensures that if $\sy$ over $G_{\bin}$ is realizable, then it is also realizable via a
term that is a complete binary tree.

\begin{lemma}
  \label{lem:complete-binary}
  Let $G_{\bin}$ be a complete binary form of a grammar $G$.
  Then for any $f \in G_{\bin}$, there also exists $f' \in G_{\bin}$ such that
  $f'$ is a complete binary tree and $\sem{f} = \sem{f'}$.
\end{lemma}

Lemma~\ref{lem:complete-binary} and Theorem~\ref{thm:bin-soundness} allow us to consider 
only complete binary trees as solutions when considering the realizability of a 
synthesis problem.

%Notice a small caveat with Theorem~\ref{thm:bin-soundness}, in that it says nothing about
%$f_{\bin}$ being a complete binary tree.
%This is because accepting only complete binary trees is outside the power of a regular grammar.
%Nevertheless, the construction of $G_{\bin}$ is strong enough to guarantee 
%that every term in $G_{\bin}$ also has an equivalent complete binary representation, 
%which ensures that if $\sy$ over $G_{\bin}$ is realizable, then it is also realizable via a 
%term that is a complete binary tree.
%
%\begin{lemma}
%  \label{lem:complete-binary}
%  Let $G_{\bin}$ be a complete binary form of a grammar $G$.
%  Then for any $f \in G_{\bin}$, there also exists $f' \in G_{\bin}$ such that 
%  $f'$ is a complete binary tree and $\sem{f} = \sem{f'}$.
%\end{lemma}

We wrap this section up with a recap of why we exactly introduced complete 
binary trees as a solution for dealing with syntax: encoding terms using the 
ideas machinery in this 
section will greatly simplify encoding the semantics $\sem{\cdot}$ as a first-order 
formula in \S\ref{SubSe:semantics}.
It is true that there are many other, perhaps simpler,
solutions for checking that $f$ is syntactically correct; 
for example, an alternative approach could be to directly encode 
the automaton for a grammar $G$ as a formula instead.
%In this sense, complete binary trees may seem a non-standard solution to an 
%already well-known problem.
However, it is difficult, or perhaps unintuitive, to generalize automata towards 
encoding the computation of a semantics, 
while the idea of manipulating trees that embed sequences directly as part of a formula 
will allow for a very natural encoding of semantics in \S\ref{SubSe:semantics}.
In particular, we will rely extensively on the fact that the complete binary 
tree encoding allows one to compute the \emph{indices} of a children node from 
given the index of a parent node, which will allow us to directly mimic 
the computation of recursive big-step semantics using a formula.

%TODO: Change defeqs to equiv? Both do make sense...

%% file: 4-2semantics.tex
\subsection{Constructing $\sem{\cdot}$ for a Loop-Free Fragment of $\ltgt$}
\label{SubSe:semantics}

Before constructing $\sem{\cdot}$ for the entirety of terms in $\ltgt$, let us first 
illustrate the key ideas for constructing $\sem{\cdot}$ for a loop-free fragment of 
$\ltgt$.
Starting with the loop-free fragment will provide a much more clear picture of the 
overall encoding; later in \S\ref{SubSe:loops}, we will extend the ideas presented in 
this section towards loops as well.

\myparvs{Value Trees for Computing Semantics}
The construction of $\sem{\cdot}$ relies on the intuition that one may construct a tree 
that encodes the execution of a program, 
similar to how in
\eqref{while-sem} from 
Lemma~\ref{lem:while-sem}, 
the sequence $\sigma_0, \cdots, \sigma_k$ encodes the execution of a while loop.
Intuitively, instead of a sequence of states as in \eqref{while-sem}, 
we rely on a \emph{tree of values} with a structure identical to the term being 
evaluated (which in turn may be encoded as a formula by relying on what we developed in 
\S\ref{SubSe:syntax}).
Example~\ref{ex:semantic-tree} illustrates an example of how such a value tree 
would be constructed for a simple term.

\begin{example}
  \label{ex:semantic-tree}
  Recall the complete binary representation of the term $1 + x + 1$
  from Example~\ref{ex:complete-binary-tree}.
  Diagram~\ref{Eq:diag_2} illustrates the syntax tree for the 
  complete binary term (on the left) and the corresponding value tree (on the right) 
  for an input state $\set{x = 3}$.

  \begin{minipage}{0.9\linewidth}
    \begin{center}
    $\quad$
    \begin{tabular}{cc}
%      \Tree[.$+\ (6)$ [.$+\ (4)$ [.$1\ (1)$ $\nullleaf\ (\nullval)$ $\nullleaf\ (\nullval)$ ] 
%                                 [.$x\ (3)$ $\nullleaf\ (\nullval)$ $\nullleaf\ (\nullval)$ ] ] 
%                      [.$1\ (1)$ [.$\nullop\ (\nullval)$ $\nullleaf\ (\nullval)$ $\nullleaf\ (\nullval)$ ] 
%                                 [.$\nullop\ (\nullval)$ $\nullleaf\ (\nullval)$ $\nullleaf\ (\nullval)$ ] ] ]
 
    \begin{tikzpicture}[scale=.8]
      \Tree[.$+$ [.$+$ [.$1$ $\nullleaf$ $\nullleaf$ ] 
                                 [.$x$ $\nullleaf$ $\nullleaf$ ] ] 
                       [.$1$ [.$\nullop$ $\nullleaf$ $\nullleaf$ ] 
                                 [.$\nullop$ $\nullleaf$ $\nullleaf$ ] ] ]
    \end{tikzpicture}
      &
    \begin{tikzpicture}[scale=.8]
      \Tree[.$6$ [.$4$ [.$1$ $\nullval$ $\nullval$ ] 
                                 [.$3$ $\nullval$ $\nullval$ ] ] 
                       [.$1$ [.$\nullval$ $\nullval$ $\nullval$ ] 
                                 [.$\nullval$ $\nullval$ $\nullval$ ] ] ]
    \end{tikzpicture}
    \end{tabular}
    \end{center}
  \end{minipage}
  \begin{minipage}{0.05\linewidth}
    \begin{equation}
      \label{Eq:diag_2}
    \end{equation}
  \end{minipage}

  Observe how the value tree contains the values that one would obtain 
  by evaluating each subexpression starting at the root of the 
  syntax tree.
  For example, the second node (for the lower-left $+$ operator in the syntax tree) 
  in the value tree is $4$, as $+$ adds $1$ and $3$ (the values from its children nodes) 
  and adds them to produce $4$.
\end{example}

In essence, value trees 
encode a bottom-up computation of a specific term as a tree.
Having a fixed representation of the computation then allows us to then apply 
a `guess-and-check' system as a formula 
similar to Lemma~\ref{lem:while-sem}, where 
one can check that each node has computed the correct value 
according to its operator.

\begin{example}
  \label{ex:value-tree}
  Reconsider the value tree from Example~\ref{ex:semantic-tree}, which may be represented as 
  $\lrangle{6, \mymathcolor{dgreen}{4}, 1, 
    \mymathcolor{cyan}{1}, \mymathcolor{cyan}{3}, \nullval, \nullval, \nullval, \nullval, \nullval, 
    \nullval, \nullval, \nullval, \nullval, \nullval}$ as a preorder-traversal sequence 
    (values to be referenced are color-coded).
  An encoding of $\sem{\cdot}$ as a first-order formula would then 
  check whether each value in this sequence is identical to the values obtained by 
  evaluating $1 + x + 1$ in a bottom-up fashion.

  For example, 
  the check at $\mymathcolor{dgreen}{t_1}$, which has $+$ as its operator in the syntax tree, 
  checks if $\mymathcolor{dgreen}{t_1} = \mymathcolor{cyan}{t_3} + \mymathcolor{cyan}{t_4}$ 
  ($t_i$ indicates the $i$-th node).
  Swapping in the values 
  from the sequence, we obtain $\mymathcolor{dgreen}{4} = \mymathcolor{cyan}{1} + \mymathcolor{cyan}{3}$ 
  (which is clearly true, as this value tree is correct).

  On the other hand, suppose that the value tree was malformed, represented by the sequence 
  $\lrangle{6, \mymathcolor{dred}{5}, 1, 
    \mymathcolor{cyan}{1}, \mymathcolor{cyan}{3}, \nullval, \nullval, \nullval, \nullval, \nullval, 
    \nullval, \nullval, \nullval, \nullval, \nullval}$ instead 
  (the value indicated in {\color{dred}red} has changed from $\mymathcolor{dgreen}{4}$ to 
   $\mymathcolor{dred}{5}$).
  In this case, the node checks whether 
  $\mymathcolor{dred}{5} = \mymathcolor{cyan}{1} + \mymathcolor{cyan}{3}$ instead, 
  which is clearly false---meaning this value tree is wrong, and that a correct 
  encoding of $\sem{\cdot}$ as a formula should reject such value trees.
\end{example}

\mypar{Sequential Composition}
One operator for which it may not be immediately clear how to check a bottom-up 
computation is sequential composition ($\Eseq{s_1}{s_2}$), 
which is iterative and often not 
computed by composing the results of $s_1$ and $s_2$.
This challenge can be solved by extending the value tree to contain a \emph{pair} of values 
$\sigma^{\stin}$ and $\sigma^{\stout}$, where $\sigma^{\stin}$ captures the 
\emph{input} state to a node and $\sigma^{\stout}$ captures the output state / value.
Then the semantics of sequential composition can be captured by a formula 
as illustrated in Example~\ref{ex:sequential}.

\begin{example}
  \label{ex:sequential}
  Consider a fragment of a value tree where the 
  {parent} node is a sequential composition 
  $\Eseq{s_1}{s_2}$, 
  containing the pair of states 
  $(\sigma^{\stin}_p, \sigma^{\stout}_p) = \mymathcolor{dgreen}{(\set{x = 3}, \set{x = 5})}$.
  Assume that the first child $s_1$ contains 
  $(\sigma^{\stin}_{\cone}, \sigma^{\stout}_{\cone}) = \mymathcolor{cyan}{(\set{x = 3}, \set{x = 4})}$ 
  in the value tree, 
  while the second child $s_2$ contains 
  $(\sigma^{\stin}_{\ctwo}, \sigma^{\stout}_{\ctwo}) = \mymathcolor{dred}{(\set{x = 4}, \set{x = 5})}$ 
  (values to be compared are color-coded for presentation).

  A formula can check whether this value tree is valid with respect to the semantics of 
  sequential composition by performing the following three checks:
  \begin{itemize}
    \item $\sigma^{\stin}_p == \sigma^{\stin}_{\cone}$, i.e., 
      $\mymathcolor{dgreen}{\set{x = 3}} == \mymathcolor{cyan}{\set{x = 3}}$: 
      Does the input state of $\Eseq{s_1}{s_2}$ match the input state of $s_1$?
    \item $\sigma^{\stout}_{\cone} == \sigma^{\stin}_{\ctwo}$, i.e., 
      $\mymathcolor{cyan}{\set{x = 4}} == \mymathcolor{dred}{\set{x = 4}}$: 
      Does the output state of $s_1$ match the input state of $s_2$?
    \item $\sigma^{\stout}_{\ctwo} == \sigma^{\stout}_{p}$, i.e., 
      $\mymathcolor{dred}{\set{x = 5}} == \mymathcolor{dgreen}{\set{x = 5}}$: 
      Does the output state of ${s_2}$ match the output state of $\Eseq{s_1}{s_2}$?
  \end{itemize}
\end{example}
One can clearly see that the three checks in Example~\ref{ex:sequential} 
model the semantics of sequential composition, the semantics of which are commonly given as 
$\sem{\Eseq{s_1}{s_2}}(\sigma) = \sem{s_2}(\sem{s_1}(\sigma))$ as well.

\myparvs{Constructing a First-Order Representation of $\sem{\cdot}$}
Having developed the key ideas for encoding $\sem{\cdot}$ as a first-order formula, 
we move to illustrating at a high level how \eqref{synth-def} (the second-order definition 
of program synthesis) can be reconstructed as a first-order formula.

Following Theorem~\ref{thm:bin-soundness}, we will assume that our synthesis problem 
$\sy$ is defined over a grammar in complete binary form. 
We will thus assume that $f$ is also a complete binary term, and represent $f$ 
with a preorder traversal of its syntax tree
$\lrangle{p_0, \cdots, p_{\h}}$, where $h$ is the height of the tree.
We have already established that finite sequences may be encoded as integers; 
so we further compress the sequence as a pair $(a_p, b_p)$.
This gives us a formula of the form:
{\small
\begin{equation}
  \label{Eq:synth-term-replaced}
  \sy \defeq \exists a_p, b_p. \forall \sigma \in D. (a_p, b_p) \in G \wedge 
  \phi(\sigma, a_p, b_p, \sem{\cdot}(a_p, b_p)(\sigma))
\end{equation}
}
We drop the domain $D$ for simplicity, remove the syntax check 
$(a_p, b_p) \in G$, which may be encoded as part of the semantics $\sem{\cdot}$.

We now introduce the value tree, also as a pair of integers $(a_v, b_v)$.
The value tree must differ for each input (as it essentially represents the 
computation taking place for each input) and thus we introduce it as an existential 
behind the input $\sigma$.
Observe that the value tree also contains the output value of $f$, 
i.e., $\sem{f}(\sigma)$ in \eqref{synth-def}, so we hoist $\sem{\cdot}$ out of 
$\phi$ to simplify the formula and obtain:
{\small
\begin{equation}
  \label{Eq:synth-output-replaced}
  \sy \defeq \exists a_p, b_p. \forall \sigma. \exists a_v, b_v.
    \sem{\cdot}(a_p, b_p)(\sigma)(a_v, b_v) \wedge \phi(\sigma, a_p, b_p, a_v, b_v)
\end{equation}
}
In \eqref{synth-term-replaced} and \eqref{synth-output-replaced}, we assume that 
references to specific parts of trees that are represented by a pair of integers 
(e.g, $\phi$ referencing $\sem{\cdot}(a_p, b_p)(\sigma)$, which is the root node of the value tree 
represented by $(a_v, b_v)$) are appropriately encoded following Lemma~\ref{lem:beta}.

What is left then, is to construct a formula for $\sem{\cdot}$ that checks whether 
$(a_v, b_v)$ represents a valid value tree according to the semantics of the 
syntax tree $(a_p, b_p)$.
We have already introduced how to perform this check on a local per-node basis; 
extending this check is easy thanks to the complete binary representation, 
because one may compute the indices of children nodes from the index of the parent node.
This allows us to simply iterate over the nodes in the tree while performing the check 
using a bounded quantifier, as in Lemma~\ref{lem:while-sem}.
\eqref{sem-def} captures this idea on a high level.\footnote{
  Strictly speaking, \eqref{sem-def} should also check that the nodes corresponding to 
  input values in the value tree are equal to the input state $\sigma$ in 
  \eqref{synth-output-replaced}, and also the syntax check $(a_p, b_p) \in G$; 
  these checks are straightforward and thus omitted for brevity.
}
{\small
\begin{equation}
  \label{Eq:sem-def}
  \begin{split}
    \sem{\cdot}(a_p, b_p)(\sigma)(a_v, b_v) \defeq &
    \forall i, 0 \leq i \leq 2^{h - 1} - 2. \semcheckl(p_i, v_i, v_{2i + 1}, v_{2i + 2}) \wedge \\
  & \forall i, 2^{h - 1} - 1 \leq i \leq \h. \leafcheckl(p_i, v_i)
  \end{split}
\end{equation}
}
In \eqref{sem-def}, $h$ represents the height of the syntax / value trees, 
$t_i$ the $i$-th node of the syntax tree, and $v_i$ the $i$-th node of the value tree 
(all of which can be recovered from $(a_p, b_p)$ and $(a_v, b_v)$ by Lemma~\ref{lem:beta}).
$\semcheckl$ encodes the local check we have developed in this section, while 
$\leafcheckl$ performs a similar check for leaf nodes (which do not have children nodes); 
the first line of \eqref{sem-def} simply checks that the value tree is correct
for non-leaf nodes while 
the second line checks correctness for leaf nodes.

\eqref{sem-def} correctly encodes the operation of the semantics function $\sem{\cdot}$, 
in the sense that $\sem{\cdot}(a_p, b_p)(a_v, b_v)$ will evaluate to $\Etrue$ if and 
only if $(a_p, b_p)$ and $(a_v, b_v)$ respectively encode a function $f$ and an input-output pair 
$(\sigma, \pi)$ such that $\sem{f}(\sigma) = \pi$.
Because $\sem{\cdot}$ as defined in \eqref{sem-def} is a
first-order formula in the standard model of arithmetic, it follows that \eqref{synth-output-replaced} 
is also a first-order formula: one that is equivalent to the definition of synthesis 
as in \eqref{synth-def}.

%% file: 4-3loops.tex
\subsection{Extending Value Tree to Support Loops}
\label{SubSe:loops}

Having constructed a first-order formula for synthesis problems over loop-free 
languages in \S\ref{SubSe:semantics}, we now proceed to show that 
the idea of the value tree and local checks can be extended to support 
loops as well---thus allowing \eqref{synth-output-replaced} to encode the 
full range of synthesis problems as a formula.

\myparvs{Nested Sequences in the Value Tree for Loops}
To understand how to perform a local parent-child check for the semantics of loops, 
consider an example illustrated in Figure~\ref{fig:state-transition-naive}, 
which has a parent node $\Ewhile{b_1}{s_1}$ that loops twice.
To make the presentation simpler, we will temporarily assume that 
while loops loop nondeterministically instead of having a 
loop guard; this will allow us to consider nested loops with more ease.
We first illustrate how the semantics of loops may be checked in a bottom-up fashion 
as in \eqref{sem-def} by allowing the value tree to contain 
sequences of values.

From the semantics of loops encoded as a formula as in \eqref{while-sem}, 
the parent node ($\Ewhile{b_1}{s_1}$) should at least 
contain the sequence of states produced by iterating through the loop 
$\lrangle{\sigma_0, \sigma_1, \sigma_2}$,
as illustrated in Figure~\ref{fig:state-transition-naive}.
One way to view this sequence, following our previous input-output pair intuition, 
is that the ends of the sequence $(\sigma_0, \sigma_2)$ encode the input-output pair 
of $\Ewhile{b_1}{s_1}$---indeed, the semantics of a 
loop repeating twice will return $\sigma_2$ given $\sigma_0$ as input---and that 
the remaining state $\sigma_1$ is an intermediate state required to 
validate the input-output pair $(\sigma_0, \sigma_2)$ as correct.

For the parent node $\Ewhile{b_1}{s_1}$ to validate that 
the sequence $\lrangle{\sigma_0, \sigma_1, \sigma_2}$ is correct 
according to the semantics of the loop body $s_1$, the 
child node for $s_1$
now must contain a sequence of \emph{input-output pairs}, one 
for each transition of the parent: say, 
$\lrangle{(\sigma^{\stin}_0, \sigma^{\stout}_0), (\sigma^{\stin}_1, \sigma^{\stout}_1)}$.
Then the parent node can check whether:
\begin{itemize}
  \item $\sigma_0 = \sigma^{\stin}_0 \wedge \sigma_1 = \sigma^{\stout}_0$: i.e., 
    does $\sem{s_1}(\sigma_0) = \sigma_1$?
  \item $\sigma_1 = \sigma^{\stin}_1 \wedge \sigma_2 = \sigma^{\stout}_1$: i.e., 
    does $\sem{s_1}(\sigma_1) = \sigma_2$?
\end{itemize}
The value tree in Figure~\ref{fig:state-transition-naive} contains a value tree 
for which this check evaluates to $\Etrue$.
Extending this check to sequences of arbitrary length, one can see that 
this captures exactly the iterative check in line 2 of 
\eqref{while-sem} of Lemma~\ref{lem:while-sem} (modulo the branch condition), 
showing that the semantics of loops also can be checked in a bottom-up fashion 
given that the value tree contains enough information.

\begin{figure*}
  \begin{subfigure}{0.47\linewidth}
    \includegraphics[width=\linewidth]{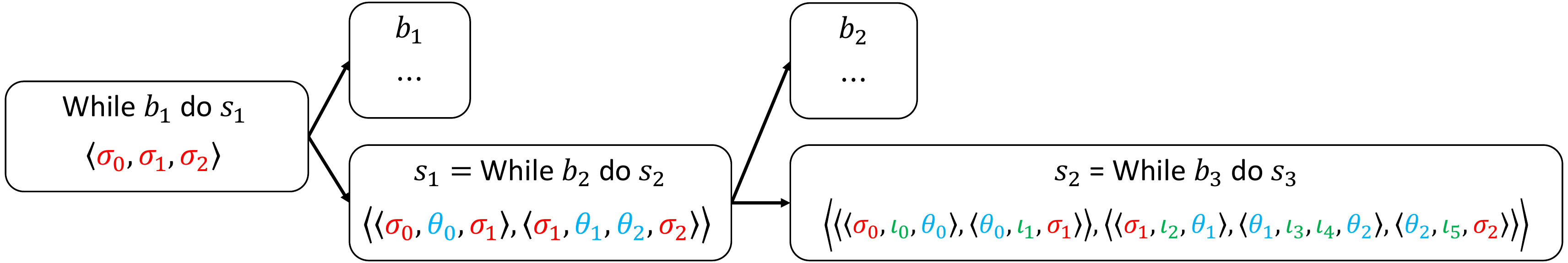}
    \caption{An example value tree considering a term with nested loops, containing nested sequences of states.}
    \label{fig:state-transition-naive}
  \end{subfigure} \hspace{0.05\linewidth}
  \begin{subfigure}{0.47\linewidth}
    \includegraphics[width=\linewidth]{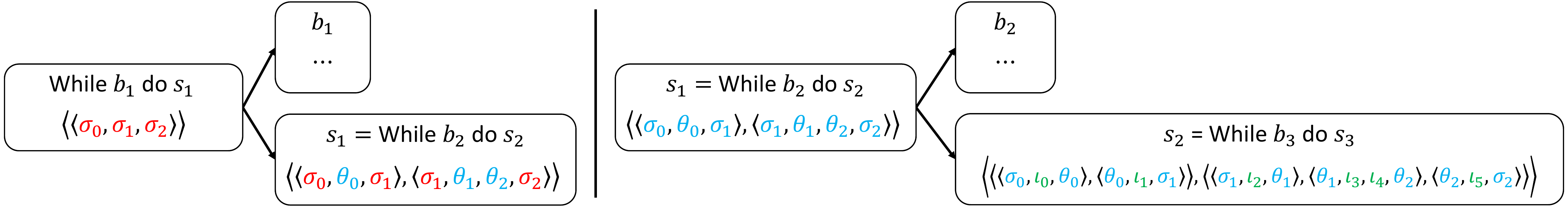}
    \caption{A local parent-child view of the value tree illustrated on the left, where the connection 
    between parent and grandchild is lost.}
    \label{fig:state-transition-answer}
  \end{subfigure}
  \caption{
    Value trees that illustrate the values required for checking the semantics of loops in a bottom-up fashion.
    States are color-coded with respect to where they originate: in Figure~\ref{fig:state-transition-naive}, note how 
    the local parent-child view for $s_1$ and $s_2$ 
    does not provide the information that, e.g., $\sigma_0$ originated from the grandparent $\Ewhile{b_1}{s_1}$.
  }
  \label{fig:state-transition-full}
\end{figure*}

However, things get more complicated when the child $s_1$ is itself a loop $\Ewhile{b_2}{s_2}$, 
as in Figure~\ref{fig:state-transition-naive}.
If $s_1$ is a loop, $s_1$ itself must validate its transitions through 
the use of intermediate states---for example, in Figure~\ref{fig:state-transition-naive}, 
the child $s_1$ itself is a loop that loops twice on the input $\sigma_0$.
Then like parent, $s_1$ will itself 
require an intermediate state (denoted $\theta_0$) 
to check the semantics of loops on $\sigma_0$, 
and also other intermediate states $\theta_1, \theta_2$, provided $s_1$
loops on $\sigma_1$ as well.

To support such scenarios in general, the value tree must contain nested sequences of 
arbitrary depth; it then becomes possible to check the semantics of loops also 
in a bottom-up fashion as previously described.
Unfortunately, 
such a variable nesting of sequences is difficult to encode as a formula, 
at least using only the $\beta$-function: nested applications of 
the $\beta$-function can encode nested sequences, but the problem is that 
the number of nestings, i.e., the number of nested applications, 
depends on the variable $l$.

\myparvs{2-nested Sequences in the Value Tree for Loops}
To fix this problem, while preserving the intuition of relying on 
sequences of states to check the semantics of loops, 
we will modify our approach to use nested sequences of maximum depth 2 
(that is, sequences-of-sequences-of-states) instead.
The key intuition is that, while one does require a nested structure to check 
that a parent loop is correctly iterating over a child body, 
this structure does not need to be preserved between parent and \emph{grandchild}.

To see this, consider Figure~\ref{fig:state-transition-answer}, where 
the example from Figure~\ref{fig:state-transition-naive} is split into a 
parent-child and child-grandchild view.

In the left of Figure~\ref{fig:state-transition-answer}, 
one can observe that a parent-child relation does require a doubly nested 
sequence of states: 
for the parent node $\Ewhile{b_1}{s_1}$ to be able to 
check the transition, e.g., $\sigma_1$ to $\sigma_2$,
the parent $\Ewhile{b_1}{s_1}$ must be able to index 
$\sigma_1$ and $\sigma_2$ in the child node.
%when constructing the local reasoning pattern for $\semcheck{G_{\bin}}$ 
%(similar to how a parent node required the ability to index a child node in \S\ref{SubSe:semantics}).
However, because the child $s_1$ may also require an arbitrarily long sequence to loop from 
$\sigma_1$ to $\sigma_2$, the nesting is required to provide \emph{structure} that 
the parent $\Ewhile{b_1}{s_1}$ can rely on to find the transition from $\sigma_1$ to $\sigma_2$ in 
the child node for $s_1$.

However, such a structure need not be preserved between parent and grandchild---the 
right of Figure~\ref{fig:state-transition-answer} illustrates how a grandchild $s_2$
needs not check a parent transition, 
e.g., from $\sigma_1$ to $\sigma_2$.
Instead, $s_2$ is only interested in validating the transitions 
that the child $s_1$ requires: e.g., $\sigma_1$ to $\theta_1$, $\theta_1$ to $\theta_2$, and 
$\theta_2$ to $\sigma_2$, where $s_2$ actually does not care whether 
$\sigma_1$ and $\sigma_2$ originated from the node for $\Ewhile{b_1}{s_1}$ or the 
node for $s_1$.
From the perspective of $s_1$, that the node for $s_2$ provides enough information 
to check its own transitions---without knowledge of whether, e.g., $\sigma_1$ 
originated from the parent $\Ewhile{b_1}{s_1}$ or not---is enough for itself 
to check the transition from $\sigma_1$ to $\sigma_2$.

Following this idea, it thus suffices that nodes in the value tree contain two-nested 
sequences of state, each of which satisfy the following local parent-child invariants: 
\begin{itemize}
  \item Each \emph{inner} sequence $\sigma_0, \cdots, \sigma_l$ in a node $s$ indicates 
    that $\sem{s}(\sigma_0) = \sigma_l$.
    If $s$ is not a loop, then $l = 1$; if $s$ is a loop, then the sequence will contain 
    the intermediate states for checking the semantics of the loop.
  \item The length of the \emph{outer} sequence of a child is identical to the number 
    of transitions in its parent.
\end{itemize}
Based on these invariants, in the first-order representation of $\sem{\cdot}$, 
a parent node $s$ checks if the $i$-th transition $\sigma_{i - 1} \rightarrow \sigma_i$ 
is correct by checking that the $i$-th subsequence in the child node starts with 
$\sigma_{i - 1}$ and ends with $\sigma_i$.
Recursively performing this check then allows us to check the semantics of 
loops in a bottom-up fashion as in \S\ref{SubSe:semantics},
thereby extending the construction of $\sem{\cdot}$ towards loops as well---and 
completing the first-order construction of \eqref{synth-def}.

%The actual formula for $\sem{\cdot}$, that accounts for the presence of loops, 
%is highly complex due to the fact that one must compute indices to reference 
%in the nested sequences.
%Because an intuitive understanding that $\sem{\cdot}$ can also be encoded as a 
%first-order formula for the case of loops suffices, 
%we refer the reader to the Appendix for the actual encoding of loops as part of $\sem{\cdot}$.

%% file: 5completeness.tex
\section{Program Synthesis is $\Sigma_{3}^{0}$-Complete}
\label{Se:completeness}

%TODO: Define synthesis as a problem

Having established in \S\ref{Se:fo}
that program synthesis can 
be expressed as a first-order formula, we now answer the 
main question of our paper: 
how exactly \emph{hard} is program synthesis?

We start by considering program synthesis as a \emph{problem}, as 
defined in Definition~\ref{def:problem}.
In this section, we will fix the grammar component of a synthesis problem 
to $\gtgt$ for simplicity; this restriction will have no effect on the 
proofs in this section.
Consider \eqref{synth-output-replaced} again, where this time 
we further compress the syntax tree $(a_p, b_p)$ as $\ptree$, and the 
value tree $(a_v, b_v)$ as $\vtree$ for simplicity:
\begin{equation}
  \sy \defeq \exists \ptree. \forall \sigma. \exists \vtree.
    \sem{\cdot}(\ptree)(\sigma)(\vtree) \wedge \phi(\sigma, \ptree, \vtree)
  \tag{\ref{Eq:synth-output-replaced}}
\end{equation}
Because the only free variable in \eqref{synth-output-replaced} is 
$\phi$, the set of solutions to \eqref{synth-output-replaced}
is exactly the set 
of $\phi$ for which $\sy$ is \emph{realizable} (i.e., has a solution).
This set captures exactly the definition of program synthesis as 
a problem:
a set of integers that encode specifications 
that have a solution to them, 
much like how the Halting problem as a problem denotes the set of 
integers that encode Turing machine-input pairs that halt.
Phrased as a decision problem, in a form likely more familiar to 
readers well-versed in synthesis, 
we obtain the question asked at the start of this paper:

\begin{center}
\emph{
"How hard is it to determine whether
  a given synthesis problem is realizable or not?"}
\end{center}

Our construction of \eqref{synth-output-replaced} now gives us 
a good opportunity to study this problem, with respect to 
the arithmetical hierarchy introduced in \S\ref{Se:prelim}.

To begin, consider \eqref{synth-output-replaced}, which is 
headed by three alternating quantifiers, followed by the 
formula body.
In the construction of $\sem{\cdot}$ in \S\ref{Se:fo}, we 
only introduced bounded universal quantifiers to iterate 
over the value tree, without introducing any unbounded quantifiers.
In addition, as discussed in \S\ref{Se:prelim} and \S\ref{Se:fo}, 
we will assume for the time being that $\phi$ is a primitive 
recursive formula.
This makes the number of unbounded quantifiers in \eqref{synth-output-replaced} 
to at most three, 
which in turn places program synthesis as $\sic{3}$ in the arithmetical hierarchy.

%While this fact in itself may seem rather underwhelming, let us
%take note of the fact that the majority of synthesis problems 
%being considered within current program-synthesis literature have 
%specifications that are \emph{quantifier-free}, 
%or contain at most \emph{one existential quantifier} 
%(with zero universal quantifiers).
%We observe that this is likely not just a random coincidence: 
%Theorem~\ref{thm:ah} tells us that
%the set of 
%specifications that may be expressed via at most one existential quantifier 
%corresponds exactly to the set where it is viable for a Turing machine 
%to at least always accept a correct input-output pair.
%The goal of program synthesis is to synthesize programs---i.e., 
%Turing machines---and thus it is only natural that the specification itself
%also be within the capabilities of a Turing machine.
%We will thus assume $\phi$ to be a Boolean formula with at most one existential 
%quantifier (i.e., a $\sic{1}$ formula) for the remainder of this section, and 
%study the computability of \eqref{synth-output-replaced} under this assumption.
%
%That $\phi$ is in $\sic{1}$ immediately has a powerful effect on the computability of 
%\eqref{synth-output-replaced}, as the number of unbounded alternating quantifiers in 
%\eqref{synth-output-replaced} is now fixed to three.
%Thus it follows that synthesis problems are at most $\sic{3}$ in the arithmetical 
%hierarchy (Theorem~\ref{thm:sic-3}), as \eqref{synth-output-replaced} 
%serves as witness of a formula that encodes programs synthesis 
%using only three alternating quantifiers.

\begin{theorem}[Program Synthesis is in $\sic{3}$]
  \label{thm:sic-3}
  Let $\syset$ be the set of realizable 
  synthesis problems.
  Then $\syset$ is in $\sic{3}$.
\end{theorem}

Theorem~\ref{thm:sic-3} provides us with an upper bound on the hardness of synthesis problems, 
but does not provide a \emph{lower} bound on hardness
(i.e., with only Theorem~\ref{thm:sic-3}, it may very well be that 
there exists an alternative first-order characterization of synthesis using a fewer number of quantifiers).
In Theorem~\ref{thm:sic-3-hard}, we prove that the lower bound of 
program synthesis is also $\sic{3}$, by proving that 
synthesis is $\sic{3}$\emph{-hard} 
(i.e., any problem in $\sic{3}$ may be reduced to an instance of a synthesis problem).

\begin{theorem}[Program Synthesis is $\sic{3}$-Hard]
  \label{thm:sic-3-hard}
  Let $\syset$ be the set of realizable 
  synthesis problems.
  Then $\syset$ is $\sic{3}$-hard.
\end{theorem}

\begin{proof}
  To prove that program synthesis is $\sic{3}$-hard, we will rely on the fact that 
  $\cof$, the set of functions which halt on a co-finite set of inputs 
  (i.e., the set of functions which do not terminate only for a finite set of inputs) 
  is $\sic{3}$-complete (and thus also $\sic{3}$-hard), and 
  reduce the decision problem for $\cof$ into a decision problem for $\syset$.

  \begin{definition}[The Set COF]
    \label{def:COF}
    Let $\hset(g)$ denote the set of inputs that halt for a function
    $g \in \ltgt$, and $\overline{A}$ denote the complement of a set $A$.
    $\cof$, the set of all co-finite functions, is defined as following:
    \[
      \cof \defeq \set{g \mid \overline{\hset(g)} \text{ is finite}}
    \]
  \end{definition}
  Definition~\ref{def:COF} is typically stated for Turing machines as opposed to 
  terms from a language.
  In this paper, we take 
  advantage of the fact that $\ltgt$ is Turing-complete in order to give a 
  alternative language-based definition;
  this definition will be far more useful in the actual reduction proof, as the 
  reduction may now consider terms only terms from $\ltgt$ for both sides of the 
  reduction (as opposed to considering Turing machines for $\cof$  and 
  terms from $\ltgt$ for $\syset$ separately).

  As stated, we make use of the fact that $\cof$ is $\sic{3}$-complete.

  \begin{lemma}[\citet{soare}]
    $\cof$ is $\sic{3}$-complete.
  \end{lemma}

  We wish to show that any algorithm 
  capable of solving the decision problem for $\syset$ 
  is also capable of solving the decision problem for $\cof$.
  We start by characterizing the concept of \emph{halting} itself
  as a formula.

  \begin{lemma}
    \label{lem:halting-pa}
    Let $f$ be a term in $\ltgt$ and $x$ be an input to $f$.
    Then the following holds:
    \begin{equation*}
      \begin{split}
        f \textnormal{ halts on } x & 
          \leftrightarrow \exists y. \semantics{f}(x) = y \\
        f \textnormal{ does not halt on } x & 
          \leftrightarrow \forall y. \semantics{f}(x) \neq y
      \end{split}
    \end{equation*}
  \end{lemma}

  The proof of Lemma~\ref{lem:halting-pa} is simple: $f$ halts on an 
  input $x$ if and only if there exists a finite sequence of 
  intermediary states $\sigma_0, \cdots, \sigma_k$ that $f$ iterates 
  through in order to reach the final state.
  Assuming that such a sequence exists, $\sigma_k$ is a witness to the 
  validity of $\exists y. \semantics{f}(x) = y$.
  If $f$ does not halt on $x$, then there does not exist any such sequence and thus 
  $\forall y. \semantics{f}(x) \neq y$.
  A similar construction of halting as a formula may also be found in~\citet{soare}.

  Now, apply Lemma~\ref{lem:halting-pa} to Definition~\ref{def:COF}, to 
  obtain the following equivalent definition of $\cof$:
  \begin{equation}
    \label{Eq:cof-pa-def}
    \cof \equiv \set{g \mid \exists x. \forall y. y \leq x \vee \exists z. \sem{g}(y) = z}
  \end{equation}
  In \eqref{cof-pa-def}, $x$ serves as the limit on the size of inputs on which 
  $g$ may not halt, which is guaranteed to exist, as by definition the set of inputs on which 
  $g$ does not halt is finite.
  The latter condition $\exists z. \sem{g}(y) = z$ states that if $y$ is bigger than $x$ 
  (the limit), then $g$ must terminate on $y$ (producing $z$ as the output).

  %\begin{equation}
  %  \label{Eq:rec-pa-def}
  %  \cof \equiv \biglbbar \ g \ \Bigg|
  %  \begin{array}{ccc}
  %    & f \text{ is recursive } & \wedge \\
  %    \exists f. \forall x. &  ( \lnot \semantics{f}(x) \wedge \forall y_1. \semantics{g}(x) \neq y_1 & \vee \\
  %    & \quad \semantics{f}(x) \wedge \exists y_2. \semantics{g}(x) = y_2 )&
  %  \end{array}
  %  \bigrbbar
  %\end{equation}
  As $g$ is a term from $\ltgt$, we wish to rewrite \eqref{cof-pa-def} using $\sem{\cdot}$, 
  as to better further manipulate this problem into an instance of a program 
  synthesis problem.
  We will thus replace $g$ with $\gtree$, and the `output' of $g$ with a value tree 
  $\vztree$, and covert to prenex normal form in \eqref{cof-pa-def} in order to obtain \eqref{cof-sem-def}:
  \begin{equation}
    \label{Eq:cof-sem-def}
    \cof \equiv \set{\gtree \mid \exists x. \forall y. \exists \vztree. y \leq x \vee \sem{\cdot}(\gtree)(y)(\vztree)}
  \end{equation}
  Observe how \eqref{cof-sem-def} now resembles a synthesis problem in structure: 
  $x$ as the function to be synthesized, $y$ as the input, and $\vztree$ as the output.
  $\gtree$ and the predicate surrounding it should be treated as the specification, 
  not the function to be synthesized, in order to match cofinite $\gtree$ with 
  realizable specifications.
  Based on this intuition, we construct the synthesis query in 
  \eqref{cof-synth-def}:
  \begin{equation}
    \label{Eq:cof-synth-def}
    \begin{split}
      \exists \ftree. \forall y. \exists \vtree. 
        & \sem{\cdot}(\ftree)(y)(\vtree) \wedge \\
        & \semout(\vtree) = (\decode_1(\ftree), \_) \wedge \\
        & y \leq \decode_1(\ftree) \vee \sem{\cdot}(\gtree)(y)(\decode_2(\vtree))
    \end{split}                              
  \end{equation}
  In \eqref{cof-synth-def}, we introduce some new notation to simplify the presentation.
  $\semout(\vztree)$ is a predicate that unpacks the root of the 
  value tree encoded by $\vztree$, i.e., the `output value' of $f$ on $x$, while $\_$ denotes 
  an unconstrained value.
  $\decode_1$ and $\decode_2$ are predicates that `decode' an input tree according to the 
  following rules:  
  \begin{itemize}
    \item Assume that $\ftree$ represents a composed pair of programs $f_q$ and $f_z$ 
      as following: $\fpair({f_q}, \Eif{y > f_q}{f_z})$.
    \item $\fpair(f, g)$ intuitively encodes a function $h$ such that 
      $h(y) = (f(y), g(y))$ for all input $y$.\footnote{
        $(f(y), g(y))$ denotes the result of pairing $f(y)$ and $g(y)$ through a suitable 
        pairing function, e.g., the Cantor pairing function.
      }
      We observe that it is always possible to construct $\fpair$ in $\gtgt$ (which is Turing-complete).
    \item $\decode_1(\ftree)$ then returns the syntax tree $\fqtree$ for $f_q$.
    \item $\decode_2(\vtree)$ then returns the value tree $\vztree$ for $f_z$.
  \end{itemize}

  With respect to this decoding, observe the second line of \eqref{cof-synth-def}: we wish the target 
  function $f$ to essentially be composed of a \emph{quine} $f_q$, that returns itself when executed, 
  and some other arbitrary function $f_r$.
  The intuition is that the quine part of $f$ will serve as $x$ from \eqref{cof-sem-def}, i.e., the 
  upper limit for the nonterminating inputs of $g$.
  Under this intuition, we wish to prove that 
  for arbitrary $\gtree$, the synthesis problem in \eqref{cof-synth-def} is 
  realizable iff $\gtree \in \cof$ as defined in \eqref{cof-sem-def}.

  If $\gtree \in \cof$, by \eqref{cof-sem-def}, there must exist $x$ that acts as the upper limit 
  of nonterminating inputs on $g$.
  We construct a solution $\ftree$, which is the syntax tree for a function $f$, for \eqref{cof-synth-def} as following:
  \begin{itemize}
    \item Let $f_q$ be some quine such that $f_q > x$.
    \item $f$ is the program $\fpair({f_q}, \Eif{y > f_q}{g(y)})$.
  \end{itemize}
  Such a program always exists in $L(\gtgt)$ because the set of quines is infinite, and thus we can 
  always find a quine $f_q > x$.

  \begin{lemma}
    The set of quines in $L(\gtgt)$ is infinite.
  \end{lemma}
  \begin{proof}
    Start with the fact that by the Kleene recursion theorem, we are guaranteed the existence 
    of at least one quine in $L(\gtgt)$; call this quine $q$.

    Observe that $L(\gtgt) \setminus \set{q}$ is also a Turing-complete language, because $q$ is 
    a constant and there are infinitely many programs in $L(\gtgt)$ that are 
    behaviorally equivalent to $q$ (i.e., accept the same set of inputs).
    Thus it follows that $L(\gtgt) \setminus \set{q}$ also contains a quine by the Kleene recursion theorem, 
    and further, that the set of quines in $L(\gtgt)$ is infinite.
  \end{proof}

  Because we can always find an appropriate $f$ (i.e., $\ftree$) for any $\gtree \in \cof$, 
  it follows that \eqref{cof-synth-def} is realizable if $\gtree \in \cof$.

  Conversely, suppose that \eqref{cof-synth-def} is realizable, witnessed by the function $f$' 
  (with the syntax tree $\ftree'$).
  Then it follows that $\decode_1(\ftree')$ is a witness of the cofiniteness of $g$, as 
  for any $y > \decode_1(\ftree')$, $g$ is guaranteed to terminate as witnessed by the value tree
  $\decode_2(\vtree)$.
  Note that in this case, the structure of $f'$ does not matter---it suffices that there exists 
  an $f'$ that satisfies \eqref{cof-synth-def}.

  Thus it follows that \eqref{cof-sem-def} and \eqref{cof-synth-def} have the same set of solutions, 
  meaning that the decision problem for \eqref{cof-sem-def} (i.e., membership in $\cof$) can be 
  reduced into an instance of the decision problem for program synthesis.
  Thus program synthesis is $\sic{3}$-hard.
\end{proof}

By Theorems~\ref{thm:sic-3} and \ref{thm:sic-3-hard}, it follows that program synthesis 
is $\sic{3}$-complete, finalizing the classification of synthesis within the 
arithmetical hierarchy.

\begin{theorem}[Program Synthesis is $\sic{3}$-Complete]
  \label{thm:sic-3-complete}
  Let $\syset$ be the set of realizable 
  synthesis problems.
  Then $\syset$ is $\sic{3}$-complete.
\end{theorem}

We observe that, because we proved that $\syset$ is $\sic{3}$-complete for the fixed 
grammar $\gtgt$, while \S\ref{Se:fo} shows that synthesis is in $\sic{3}$ for any grammar, 
it follows that synthesis is still $\sic{3}$-complete even when the grammar is not 
fixed to $\gtgt$.
In other words, the ability to choose a grammar has \emph{no effect} on the 
computational hardness of program synthesis.

That $\syset$ is $\sic{3}$-complete means that, having access to an oracle for 
program verification (which is $\pic{2}$-complete for primitive recursive 
specifications, as discussed in \S\ref{Se:prelim}) makes $\syset$ 
recursively enumerable.
Size-based enumeration algorithms are a good witness of this fact: 
such algorithms will always terminate for realizable problems, given 
a verification oracle that can check whether a specific program satisfies 
the desired property.

%will do just this: 
%because a term satisfying a specification is guaranteed to be of finite size, 
%naive enumeration is guaranteed to terminate for realizable synthesis problems.

However, the $\sic{3}$-completeness of $\syset$ also means that one 
\emph{cannot do better} than the enumerative algorithm in terms of computability.
In particular, it means that there \emph{cannot} 
exist an algorithm capable of \emph{rejecting unrealizable} synthesis problems 
within a finite number of steps, even with access to an oracle for program verification.
This, in turn, implies the non-existence 
of a complete algorithm capable of performing both synthesis and unrealizability 
at the same time.

%% file: 6variants.tex
\section{The Computability of Variants of Synthesis Problems}
\label{Se:variants}

Having established the computational hardness of program synthesis 
in general through Theorem~\ref{thm:sic-3-complete}, in this section 
we consider the computational 
hardness of many \emph{variants} of program synthesis.
By variants, we refer to both synthesis problems that have 
restrictions posed on them (often to make solving synthesis problems 
practically easier), and 
also those where the definition has been relaxed 
(in order to pose queries that are more complex).

%\myparvs{Synthesis over Complex Input Domains}
%In \S\ref{Se:fo} and \S\ref{Se:completeness}, we temporarily assumed 
%that the input domain $D$ of a synthesis problem was fixed to the 
%entire set of natural numbers $\mathbb{N}$ for simplicity. 
%
%The ability to choose an arbitrary input domain for a synthesis problem 
%has a similar effect to being able to choose the specification: 
%the synthesis query becomes harder as the membership query
%for the input domain becomes harder.
%However, for infinite input domains, 
%existing work in program synthesis
%focuses almost exclusively on input domains that can be expressed via 
%first-order formulae that are quantifier-free or have only bounded 
%quantifiers
%(as is the case for specifications).
%For these cases, program synthesis remains 
%$\sic{3}$-complete as proved in \S\ref{Se:completeness}.

\myparvs{Synthesis on Finite Examples}
One popular approach to solving synthesis problems in existing work
is the idea of 
\emph{programming-by-example}~\cite{sketch, flashfill, semgus, duet}, 
where one attempts to solve simplified versions of synthesis problems 
where 
the input domain is limited to contain only a \emph{finite number of examples}.

The limitation of the input space to a finite set greatly
reduces the complexity of program synthesis:
program synthesis is only $\sic{1}$-complete when performed 
over a finite input space.

\begin{corollary}
  \label{cor:synth-fin}
  Let $\sysetfin$ denote the set of realizable synthesis problems  
  defined over a finite input domain $D$.
  Then $\sysetfin$ is $\sic{1}$-complete.
\end{corollary}

The fact that $\sysetfin \in \sic{1}$ follows from the fact 
for a finite set of inputs $\set{\sigma_0, \cdots, \sigma_k}$,
one can simply replace the universal quantifier over the input 
$\sigma$ in \eqref{synth-output-replaced} as a conjunction as in \eqref{finite-synth}.
{\small
\begin{equation}
  \label{Eq:finite-synth}
  \begin{split}
    \sy_{\mathsf{fin}} \equiv \exists \ptree. 
      \exists \vtree_0, \cdots, \vtree_k.
      & \sem{\cdot}(\ptree)(\sigma_0)(\vtree_0) \wedge \phi(\sigma_0, \ptree, \vtree_0) \\
			& \cdots \\
      & \sem{\cdot}(\ptree)(\sigma_k)(\vtree_k) \wedge \phi(\sigma_k, \ptree, \vtree_k)
  \end{split}
\end{equation}
}
The fact that $\sysetfin$ is $\sic{1}$-hard follows from the fact that 
the set $\set{\sigma_0, \cdots, \sigma_k}$ can be any finite subset of $\mathbb{N}$.
There is no algorithm that may solve the Halting problem for arbitrary 
subsets of $\mathbb{N}$ (the problem is still $\sic{1}$-complete), 
and as \eqref{finite-synth} asks that 
$\ptree$ halts on $\set{\sigma_0, \cdots, \sigma_k}$, 
$\sysetfin$ is also $\sic{1}$-hard (and thus $\sic{1}$-complete).

%The fact that synthesis over finite examples is $\sic{1}$-complete 
%means that an algorithm with access to an 
%oracle for the Halting problem is also capable of \emph{fully} solving synthesis 
%problems, in that the algorithm will also be able to identify unrealizable problems.
%This is rather surprising, as enumerate-and-check algorithms 
%as commonly implemented in existing solvers are often 
%unable to reject unrealizable synthesis problems due to the infiniteness of the 
%search space, even with access to an oracle for the Halting problem.

%TODO: Synthesis with program verification
One kind of approach that relies on programming-by-example
are those that reduce synthesis into another task, such as 
synthesis reduced to program verification~\cite{nope} or 
Constrained Horn Clause solving~\cite{semgus}.
These approaches are sound precisely because limiting the input space 
to a finite set reduces the hardness of program synthesis
to $\sic{1}$; in other words, such a reduction is not possible for 
general synthesis problems (which are $\sic{3}$-complete).
%(e.g., program verification is both computationally 
%easier and has more practical solvers)

%This also means that such reductions are guaranteed to be incomplete in 
%general.

\myparvs{Inductive Synthesis and Generalizations}
One approach closely related to programming-by-example is 
\emph{inductive synthesis}~\cite{sketch}, where one first attempts 
to solve a simplified synthesis problem over a finite number of inputs, 
then \emph{generalize} the synthesized solution to the entire (possibly infinite) 
input space.
Inductive synthesis 
has been proven to be very effective in practice; 
Sketch~\cite{sketch}, 
Neo~\cite{conflict}, 
Duet~\cite{duet},
and Messy~\cite{semgus}, are just a few examples amongst 
the many solvers that rely on inductive synthesis in some way.

The $\sic{1}$-completeness of $\sysetfin$ allows us to state a 
corollary on the hardness of \emph{generalization} in inductive synthesis 
(i.e., the task 
of extending a function that is correct on a finite set of inputs 
to be correct on the entire, possibly infinite, input space).

%NOTE: Can be understood via the spirit of partial evaluation

\begin{corollary}
  \label{cor:generalizability}
  Let $\sy$ denote some synthesis problem over an infinite domain $D$, and 
  $\sy_d$ denote the same synthesis problem where the domain is a finite subset $d \subset D$.
  Then there \emph{cannot} exist a computable generalization algorithm that takes an arbitrary solution 
  $f_d$ for $\sy_d$ and converts it into a solution $f$ for $\sy$, 
  for arbitrary $\sy$ and $\sy_d$: generalization is \emph{uncomputable} ($\sic{2}$-complete).
\end{corollary}

Generalization plays a key role in inductive synthesis, where most approaches rely on an 
algorithm known as 
counterexample-guided inductive synthesis (CEGIS)~\cite{sketch}.
In CEGIS, one first synthesizes a candidate function $f_d$
that is correct on a finite set of examples $d$, 
then attempts to find a \emph{counterexample}: some input $x \in D$ (the full input domain)
such that $f_d$ fails to satisfy the specification on $x$.
If one succeeds in finding such an $x$, then $x$ is added to the set of examples $d$ 
and the algorithm repeats.
If not, then $f_d$ is correct on all inputs and the synthesis problem has been solved.

To the best of our knowledge, the only result on the properties of generalization 
itself is that CEGIS is simply undecidable~\cite{nay}; 
Corollary~\ref{cor:generalizability} gives us a precise result on the hardness of 
generalization algorithms in general, 
which, as discussed in \S\ref{Se:intro}, 
remain surprisingly less-studied despite their practical relevance.

In particular, 
Corollary~\ref{cor:generalizability} 
tells us that CEGIS is \emph{suboptimal} as a generalization 
algorithm in terms of computational hardness, 
despite the fact that CEGIS is well-known to be very effective in practice.
This is due to the fact that searching for a counterexample $x$ is a 
$\sic{2}$-complete problem, 
 complement of safety verification as discussed in \S\ref{Se:prelim}.
Thus it follows that if one has access to an oracle for finding counterexamples, 
one should also be able to perform generalization (generalization becomes computable).
%as both problems being 
%$\sic{2}$-complete implies that there must exist a computable translation between 
%these problems.

However, there are synthesis problems for which CEGIS will fail to compute a generalization, 
even with an oracle for finding counterexamples.
This is due to the fact that there exist synthesis problems
which are guaranteed to generate an infinite number of counterexamples 
(as illustrated in Example~\ref{ex:suboptimal}~\cite{nay}).

\begin{example}
  \label{ex:suboptimal}
  Let $\sy$ be a synthesis problem defined on the following set of programs $S$:
  \[
    S::= \Eassign{x}{E} \mid \Eif{E == y}{S} \mid \Eseq{S}{S} 
    \qquad E::= 0 \mid 1 \mid E + E
  \]
  Let the specification for $\sy$ be to synthesize a function $f$ that, for an 
  arbitrary input state $\sigma = [x \mapsto 0, y \mapsto a]$ for some $a$ and $b$, 
  $f(\sigma) = [x \mapsto a, y \mapsto a]$.

  $\sy$, when defined over a finite set of inputs $d$, will always have a solution $f_d$ 
  that chains as many $\mathsf{if}$-$\mathsf{then}$s as required.
  However, $f_d$ will always have a counterexample when generalizing the input set to 
  allow any $a \in \mathbb{N}$.
  Simply take the largest constant $C$
  that occurs in $f_d$: it is clear that $[x \mapsto 0, y \mapsto C + 1]$ is a counterexample, 
  as $x$ can never be assigned $C + 1$.

  Thus CEGIS will loop indefinitely on $\sy$, even with an oracle for finding counterexamples.
\end{example}

Of course, the existence of an algorithm that is computation-theoretically more optimal 
does not immediately mean that this algorithm will also perform better in practice.
We leave the question of finding a generalization algorithm that is 
both theoretically optimal and practically efficient as future work.

\myparvs{Synthesis for Loop-Free Languages}
Another variant of synthesis problems 
often considered are those defined over
\emph{loop-free} languages.
Loop-free languages are common in
synthesizers for specialized DSLs~\cite{flashfillplus, swizzle}; 
most background theories used with \sygus~\cite{sygus}, such as linear
integer arithmetic, or bitvectors, are also loop-free.
%Synthesizers that unroll loops up to a fixed bound, such as Sketch~\cite{sketch} 
%or Rosette~\cite{rosette}, can also be understood as having a semantics that 
%does not permit unbounded loops.

Formally speaking, in this paper we use the term
loop-free languages to refer to languages for which the semantics 
are decidable and not Turing-complete (e.g., primitive recursive languages).
For these languages, program synthesis becomes $\sic{2}$-complete.

\begin{theorem}
  \label{cor:synth-loop-free}
  Let $\sysetnoloop$ denote the set of synthesis problems where the 
  grammar $G$ is loop-free.
  Then $\sysetnoloop$ is $\sic{2}$-complete.
\end{theorem}

The fact that $\sysetnoloop$ is \emph{in} $\sic{2}$ follows from the fact that 
if $G$ is loop-free, one may remove the existential quantification over the 
value tree $\vtree$ in \eqref{synth-output-replaced}.
Because the semantics of $G$ are decidable, it is possible to encode 
$\sem{\cdot}$ directly as a quantifier-free 
primitive recursive formula that directly substitutes all occurrences 
$\vtree$ in \eqref{synth-output-replaced}.
Conversely, one can also say that the reason why the value tree must 
be existentially quantified in \eqref{synth-output-replaced} is \emph{because} 
synthesis problems, in general, may contain programs with loops.

That $\sysetnoloop$ is $\sic{2}$-hard follows from the fact that one
can reduce the decision problem for $\fin$, which is the set of functions which 
halt only on a finite set of 
inputs and well-known to be $\sic{2}$-complete, 
into a decision problem for $\sysetnoloop$.

\begin{proof}
  Let $C$ be a grammar consisting of the set of all natural numbers, 
  $C ::= 0 \mid 1 \mid C + C$.
  Consider a synthesis problem $\sy_C$ defined over $C$ as following, 
  where we have removed the value tree and replaced with direct occurrences of $\sem{\cdot}(\ftree)(x)$ 
  as discussed:
  \begin{equation}
    \label{Eq:fin-sem-reduce}
    \syc \defeq \exists \ftree. \forall x. 
      x < \sem{\cdot}(\ftree)(x) \vee \lnot \terminate(\gtree, x, \sem{\cdot}(\ftree)(x))
  \end{equation}
  In \eqref{fin-sem-reduce}, we use $\terminate(\gtree, x, \sem{\cdot}(\ftree)(x))$ to represent a predicate 
  that returns $\Etrue$ iff the term encoded by $\gtree$ terminates on the input $x$ 
  within $\ftree(x)$ steps.
  Such a predicate is clearly encodable as a $\sic{0} = \pic{0}$ formula, as 
  $\sem{\cdot}(\ftree)(x)$ limits the maximum size of the value tree that must be considered, allowing us to 
  introduce the value tree using a bounded quantifier
  (even if $g$ is not recursive).

  Then $\syc$ is a synthesis problem that is realizable iff $\gtree$ encodes a term $g$
  such that $g \in \fin$.
  If $g \in \fin$, then as $g$ terminates only on a finite set of inputs, we may take the 
  maximum number of steps $M_g$ required for $g$ to terminate on these inputs; $M_g$ serves 
  witness as a term in $L(C)$ that satisfies $\syconst$.
  On the other hand, if $\syconst$ is realizable, there must exist a $M_g$ such that 
  $g$ does not terminate on all inputs $x > M_g$; thus $g \in \fin$.

  Thus the decision problem for $\fin$ can be reduced into an instance of a loop-free
  synthesis problem, and it follows that loop-free synthesis is $\sic{2}$-complete.
\end{proof}

We observe that one cannot reduce $\cof$ into a loop-free synthesis problem 
because $g \in \cof$ can represent any general function in $\gtgt$, making it impossible 
to encode $\fpair(f_q, \Eif{y > f_q}{g(y)})$ using only loop-free languages.
Note that it does not matter even if $g$ is recursive, 
as there exist functions that are recursive but outside of any 
decidable language.

%NOTE: Synthesis by partial evaluation. Can it be done?
%Combined with the fact that many of the aforementioned synthesizers 
%employ at least a version of such enumeration to some degree, this may 
%allude to the fact that there is a more efficient procedure for 
%enumeration that remains undiscovered in our current solvers.

\myparvs{Synthesis Modulo Partial Correctness}
A similar variant to synthesis problems over loop-free languages are 
synthesis problems that only ask for \emph{partial correctness}, 
i.e., that the synthesized function need only satisfy the specification 
if the function terminates.
It is true that synthesizers often require total correctness instead 
of partial correctness, but we argue that 
at least some of this is due to the fact that, 
as previously discussed, many synthesizers 
actually focus on loop-free languages (where termination is guaranteed) 
to begin with.
%For example, synthesizers that focus on domain-specific 
%languages (DSLs)~\cite{flashfill, swizzle} often have loop-free DSLs; 
%most background theories used with \sygus~\cite{sygus}, such as linear 
%integer arithmetic, bitvectors, or strings, are also loop-free.
Synthesizers that do allow loops
often reason about loops via unrolling them up to a finite bound
(e.g, Sketch~\cite{sketch} and Rosette~\cite{rosette}), which 
may be understood as enforcing partial correctness 
up to the loop unrolling bound.

Like synthesis over loop-free languages, synthesis modulo 
partial correctness is a $\sic{2}$-complete problem.

\begin{corollary}
  \label{cor:synth-partial}
  Let $\sysetpartial$ denote the set of synthesis problems that are realizable, 
  under the condition that the specification $\phi$ 
  need only be satisfied only when the target function $f$ terminates.
  Then $\sysetpartial$ is in $\sic{2}$.
\end{corollary}

That $\sysetpartial$ is in $\sic{2}$ may be easily shown by 
updating \eqref{synth-output-replaced} to require partial correctness instead.
\[
   \sy_{\textnormal{part}} \defeq \exists \ptree. \forall \sigma. 
   (\exists \vtree. \sem{\cdot}(\ptree)(\sigma)(\vtree) \implies \phi(\sigma, \ptree, \vtree))
\]
Rewriting the implication into a disjunction 
converts the existential over the value tree $\vtree$ into a universal, 
therefore classifying program synthesis modulo partial correctness as in $\sic{2}$.

On the other hand, that $\sysetpartial$ is $\sic{2}$-hard can be proved 
in the same way as for $\sysetnoloop$, by reducing the decision problem for
$\fin$ (as all programs in $L(C)$ are guaranteed to terminate anyways).
We observe that, like the case for loop-free languages, one cannot reduce 
$\cof$ into synthesis modulo partial correctness, this time 
because the upper limit $\decode_1(\ftree)$ must exist for all $y$ 
(but will not if $f$ is allowed not to terminate on some inputs).

We expect partial correctness to be considered increasingly 
often as synthesizers expand their support for loops:
consider the fact that partial correctness as proved by 
Hoare logic plays a major role in program verification, 
which program synthesis relies on.
%Synthesizers that unroll loops up to a fixed bound, such as 
%Sketch~\cite{sketch} and Rosette~\cite{rosette}, may be understood 
%as examples of synthesizers that perform synthesis over loops and 
%require partial correctness (up to the unroll bound).

\myparvs{Synthesis with Complex Specifications}
Having discussed some easier variants of synthesis problems, 
we now consider some scenarios that are more complex 
than what we have discussed in \S\ref{Se:completeness}.

A good example is when the set of specifications is allowed to become 
stronger than decidable formulae: this may happen, for example, 
when the specification is given
as a reference implementation  (i.e., a $\sic{1}$-formula) 
as opposed to a primitive recursive formula.

From \eqref{synth-output-replaced}, it becomes clear that 
with $\sic{1}$ specifications, 
when considering 
only total correctness, synthesis remains $\sic{3}$-complete.
Perhaps more interesting is the case when considering 
partial correctness: when considering partial correctness with respect 
to a reference implementation, often the desire is that the 
functions should diverge on the same set of inputs.
In this case, synthesis remains  $\sic{3}$-complete 
instead of becoming simpler, 
the reason for this being that determining 
whether two functions terminate on the same set of inputs 
is itself a specification stronger than $\sic{1}$.

%terminate on the same set of inptus?

For even more complex specifications, from \eqref{synth-output-replaced} 
it immediately follows that 
for $\phi \in \sic{n + 1}$ or $\phi \in \pic{n}$, synthesis 
lies in $\sic{n + 3}$.

\myparvs{Synthesis with Quantitative Objectives}
An interesting variant of synthesis problems are those with 
\emph{quantitative objectives}~\cite{qsygus} over the 
syntax of the produced function: for example, to impose a 
maximum size on the solution.
The ability to express such specifications is already included 
in Definition~\ref{def:synth-def}, which allows $\phi$ to operate 
over the syntax tree $\ftree$; 
thus the $\sic{3}$-completeness of synthesis is preserved 
when considering synthesis with quantitative objectives.

\myparvs{Hyperproperties}
As a final variant of synthesis problems, we consider 
those where the specification is a \emph{hyperproperty}~\cite{hyperproperties}: 
i.e., properties where that must hold over \emph{multiple} runs 
of a program.
Examples of hyperproperties include monotonicity, or transitivity.

Hyperproperties require a relaxation to Definition~\ref{def:synth-def}, as 
properties such as monotonicity cannot be expressed by 
calling a target function $f$ only once; we must thus relax Definition~\ref{def:synth-def} 
to allow calling $f$ multiple times.
Assuming the definition allows one to call $f$ a finite number of times, one may 
see that synthesis remains $\sic{3}$-complete via \eqref{synth-output-replaced}.

%% file: 7related.tex
\section{Discussion} 
\label{Se:discussion}

In this section, we discuss the contributions of the material in this paper, 
especially the first-order construction detailed in \S\ref{Se:fo}.

%Because we have already discussed how the proofs in this paper relate to 
%existing work in program synthesis in \S\ref{Se:variants}, 
%in this section we focus on other bodies of existing work:
%\rone those on \emph{computability}, and \rtwo those on \emph{unrealizability}.

It is true that, as noted in \S\ref{Se:intro} and 
\S\ref{Se:fo}, there exist other ways to prove 
that there exist first-order representations of program synthesis, some of 
which are perhaps simpler than the full construction 
provided in \S\ref{Se:fo}.
It is also true that, from the viewpoint of computability, 
the fact that synthesis has a first-order representation is not a strictly new 
discovery: for example, proofs of the 
undecidability of the Halting problem rely on \godel numberings 
of Turing machines, and the construction of $\sem{\cdot}$
may be understood as encoding a universal Turing machine 
as a formula, both of which are topics that have been extensively studied.

However, we argue that the construction in \S\ref{Se:fo}, which intentionally 
constructs separate representations for each element in a synthesis problem, 
is highly beneficial as a \emph{tool} for studying the theoretical properties 
of program synthesis.
In particular, such a concrete yet intuitive construction helps 
identify and clarify tiny subtleties that may arise as sources of confusion 
when studying the computational hardness of synthesis programs.

As an example, we will once again consider two variants of synthesis problems 
that bring about changes to the hardness of synthesis despite 
being small changes.
If one fixes the grammar to $\gtgt$ as in \S\ref{Se:completeness}, but instead 
\emph{disallows} the specification to refer to $f$ itself 
(i.e., as in \eqref{synth-intro}), 
then program synthesis becomes $\pic{2}$-complete instead.

This fact follows from a reduction of $\tot$, the set of all formulae $\phi(x, y)$ 
such that $\forall x. \exists y. \phi(x, y)$ is true: 
it is clear that for $\phi \not \in \tot$, any synthesis problem of the form 
$\exists f. \forall x. \exists y. f(x) = y \wedge \phi(x, y)$ is unrealizable by
definition.
Conversely, if $\phi \in \tot$, then the aforementioned synthesis problem 
is realizable by an $f$ that enumerates the evaluation of $\phi(x, 0), \phi(x, 1), \cdots$ 
until it finds $y$ for which $\phi(x, y)$ holds: such an algorithm is guaranteed 
to halt by totality of $\phi$.
Because $\exists f. \forall x. \exists y. f(x) = y \wedge \phi(f(x), x)$ also 
constitutes the \emph{entire} set of synthesis queries (under the assumption that 
$\phi$ cannot reference $f$), the two sets are identical, 
making synthesis in this scenario $\pic{2}$-complete 
(as $\tot$ is $\pic{2}$-complete).

One might guess then, that the power to reference $f$ plays a vital role in 
the computational hardness of synthesis---but actually, this is not necessarily 
the case!
To see this, consider the set of synthesis problems that are defined 
over the grammar $C$ from Example~\ref{ex:suboptimal} (e.g., the set of constant values).
Program synthesis over $C$ is actually $\sic{3}$-complete again, \emph{even if} 
$\phi$ is not allowed to reference $f$---given a slight relaxation 
that the specification $\phi$ is now allowed to refer to $\sic{1}$-formulae.
The proof of this fact follows from the fact that one can once again 
reduce $\cof$ into synthesis problems with the aforementioned restrictions: 
the set of constants allows one to find the upper limit for inputs that 
may nonterminate.

These results suggest that there is a lot of subtlety hidden in 
precisely determining the computational hardness of synthesis problems, 
which is where a constructive encoding as in \S\ref{Se:fo} is beneficial.
As illustrated in \S\ref{Se:variants}, and in this section, we argue the 
constructive first-order encoding is what allows us to cleanly and 
efficiently consider the hardness of variants of synthesis problems 
we have discussed in this paper despite these subtleties.

In addition to being useful as a tool for studying synthesis from 
a theoretical perspective, we argue that the methodology developed in 
\S\ref{Se:fo} can also be useful for other tasks and proofs.
For example, the value tree can be used as a proof technique for 
the proof of completeness in unrealizability logic~\cite{ul}, where 
a key part of the proof
is to provide a precise invariant for a set of loops 
(similar to how a key part of the completeness proof in Hoare logic 
depends on being able to provide a precise invariant for single loops). 
The value tree, which may be understood as a mechanism for capturing the semantics 
of a set of programs via guess-and-check, 
provides a methodology for one to construct such an invariant.

%In addition, this paper is the first to consider 
%program synthesis directly in the context of 
%general computability
%(aside from~\citet{undecid}, which mainly foucses on 
% identifying decidable limitations to 
% \sygus problems as opposed to general computability).

%As stated in \S\ref{Se:intro}, one of the main goals of this paper 
%is to allow researchers to take the ideas in this paper and apply 
%them to their own problems, helping them to identify where 
%hardness and complexity are inevitable and where they can be avoided.
%We believe a thorough presentation as in this paper is 
%essential to this goal.

%% file: 8conclusion.tex
\section{Conclusion}
\label{Se:conclusion}
In this paper, we have constructed a first-order representation 
of program synthesis in order to prove that 
program synthesis in general is $\sic{3}$-complete.
In addition to this main result, we have also studied 
the computability of variants of program synthesis, 
such as synthesis over finite examples, 
or generalization.
We hope that this paper will be able to serve as a 
reference for future work in program synthesis, including 
theoretical results, solving procedures, or proving unrealizability.